\documentclass[11pt,a4paper]{article}
\usepackage{color}
\usepackage{amsmath,amsthm,amssymb,graphicx}
\usepackage{algorithm, algorithmic}
\usepackage{hyperref}
\usepackage{t1enc}
\usepackage[utf8]{inputenc}
\usepackage{enumerate}
\usepackage[margin=25mm]{geometry}
\newcommand{\nopath}{\mathcal{L}}
\newcommand{\struct}{$\delta$-path}

\newcommand{\Hpath}{\mathrm{pa}}
\newcommand{\Hcycle}{\mathrm{cyc}}

\newcommand{\partition}{\Pi}
\newcommand{\weakex}{{\sc NTU-PMG-WC-Exist}}
\newcommand{\weakexforb}{{\sc NTU-PMG-WC-Exist-F}}
\newcommand{\weakverif}{{\sc NTU-PMG-WC-Verif}}
\newcommand{\strongex}{{\sc NTU-PMG-SC-Exist}}
\newcommand{\strongverif}{{\sc NTU-PMG-SC-Verif}}
\newcommand{\new}{\textcolor{black}}
\newcommand{\NO}{NO}
\newcommand{\pedge}{e}

\newcommand{\Zset}{\mathbb{Z}}
\newcommand{\SC}{\mathcal{SC}}

\newcommand{\nocyc}{\mathcal{K}}
\newcommand{\altpath}{$M_0$-alternating path}
\newcommand{\altcyc}{$M_0$-alternating cycle}
\newcommand{\kstar}{\mathcal{K}^*}

\providecommand{\keywords}[1]
{
  \small	
  \textbf{\textit{Keywords---}} #1
}

\theoremstyle{plain}

\newtheorem{theorem}{Theorem}
\newtheorem{lemma}[theorem]{Lemma}
\newtheorem{cor}[theorem]{Corollary}
\newtheorem{claim}[theorem]{Claim}

\newtheorem{example}[theorem]{Example}

\theoremstyle{definition}

\newtheorem{definition}[theorem]{Definition}

\title{The NTU Partitioned Matching Game for International Kidney Exchange Programs}
\author{Gergely Csáji \thanks{E\"otv\"os Loránd University and HUN-REN KRTK KTI {\tt csaji.gergely@krtk.hun-ren.hu}} \and Tamás Király\thanks{
HUN-REN--ELTE Egerv\'ary Research Group,
Department of Operations Research, E\"otv\"os Lor\'and University, Budapest. Email: {\tt tamas.kiraly@ttk.elte.hu}}
\and Zsuzsa Mészáros-Karkus\thanks{Department of Operations Research, E\"otv\"os Lor\'and University, Budapest. Email: {\tt karkuszsuzsi@gmail.com}}
}

\date{September 2025}

\begin{document}

\maketitle

\begin{abstract}
Motivated by the real-world problem of international kidney exchange (IKEP), recent literature introduced a generalized transferable utility matching game featuring a partition of the vertex set of a graph into players, and analyzed its complexity.
We explore the non-transferable utility (NTU) variant of the game, where the utility of players is given by the number of their matched vertices. Our motivation for studying this problem is twofold. First, the NTU version is arguably a more natural model of the international kidney exchange program, as the utility of a participating country mostly depends on how many of its patients receive a kidney, which is non-transferable by nature. Second, the special case where each player has two vertices, which we call the NTU matching game with couples, is interesting in its own right and has intriguing structural properties.

We study the core of the NTU game, which suitably captures the notion of stability of an IKEP, as it precludes incentives to deviate from the proposed solution for any possible coalition of the players. 
We prove computational complexity results about the weak and strong cores under various assumptions on the players. 
In particular, we show that if every player has two vertices, then the weak core is always nonempty, and the existence of a strong core solution can be decided in polynomial time. Moreover, one can efficiently optimize on the strong core. In contrast, it is NP-hard to decide whether the strong core is empty when each player has three vertices. We also show that if the number of players is constant, then the non-emptiness of the weak and strong cores is polynomial-time decidable, and we can find a minimum-cost core solution in polynomial time.
\end{abstract}


\section{Introduction}

We study the \emph{NTU Partitioned Matching Game}, which is a cooperative game defined as follows. Let $G = (V;E)$ be a graph, and let $V_1, V_2, \dots, V_m$ be a partition of $V$ into $m$ subsets, where $m$ is the number of players. Each class of this partition corresponds to a player. The utility of a subset $\{i_1,\dots,i_k\}$ of players is defined as the maximum size of a matching in the subgraph of $G$ induced by the vertex set $V_{i_1}\cup V_{i_2} \cup \dots \cup V_{i_k}$. The transferable utility (TU) version of this game was introduced by  
Kern et al.~\cite{kern2019generalized}. The non-transferable utility (NTU) version, where the utility of a matching $M$ for player $i$ is the number of vertices of $V_i$ covered by $M$, has been considered in a technical report by a subset of the present authors \cite{egres-19-12}, as well as by Collette \cite{collette2023weak}.

In this paper, we focus on computational complexity problems related to the weak and strong cores of the NTU game. Our motivation for studying the core is twofold. First, International Kidney Exchange Programs (IKEPs) can be naturally modeled as a partitioned matching game, and considering the NTU core is an appealing way to strengthen the concept of individual rationality, which takes into account the possibility that some countries might have the incentive to abandon the program together and arrange the exchanges among themselves. The second motivation comes from theoretical considerations: the computational complexity of the core has particularly interesting properties in the case where every partition class has size at most 2, which we call the \emph{NTU matching game with couples}. In particular, we can prove that the non-emptiness of the strong core can be decided in polynomial time. The proof is quite involved and requires complex arguments from matching theory. In the following, we describe these two motivations in detail.

\subsection{Motivation: international kidney exchange programs}

Kidney exchange programs (KEPs) enable patients with incompatible live donors to receive transplants by swapping donors with other incompatible pairs. These exchanges, typically organized into 2-way or 3-way cycles, have facilitated thousands of life-saving transplants worldwide \cite{biro2019building}. To improve access and efficiency, some national programs have begun to collaborate through international kidney exchange programs (IKEPs), pooling patient-donor pairs across borders \cite{bohmig2017czech,valentin2019international}. While international coordination substantially expands the set of potential matches and can increase both the number and quality of transplants, it also introduces new challenges --- most notably, how to incentivize countries (or hospitals, \new{for example in the US setting}) to actively participate in these programs.

These incentive issues are not merely theoretical. In the US, the kidney exchange market facilitates hundreds of transplants annually, yet remains fragmented and inefficient \cite{agarwal2019market}. Hospitals often withhold easy-to-match pairs, participating only with hard-to-match patients, which undermines the overall efficiency of the system. A similar pattern appears in Europe, where IKEPs are often limited to hard-to-match pairs that were not matched during national kidney exchange programs~\cite{biro2019building}, greatly reducing their effectiveness \cite{druzsin2024performance}. Institutional fragmentation and misaligned incentives therefore remain a central obstacle to improving outcomes in both national and international kidney exchange settings.

\new{To avoid serious challenges that arise when trying to ensure strategy-proofness (in particular, the impossibility of achieving both strategy-proofness and maximality or near-maximality of the number of transplants}~\cite{ashlagi2012new,ashlagi2014free,roth2005pairwise}), in this paper we adopt the framework of cooperative game theory, modeling IKEPs as partitioned matching games where players -- typically countries or hospitals-- seek to maximize the number of their own patients who receive a transplant, similarly as in~\cite{kern2019generalized}. We focus on pairwise exchanges (2-cycles), and assume unweighted utility, where each matched patient contributes equally to the country's objective. Due to medical and logistical constraints, allowing only pairwise exchanges is adopted by many countries, such as Denmark, Norway, Sweden and France~\cite{biro2019building,biro2021modelling}, where the former three also share cooperation internationally.
The underlying matching problem is represented by an undirected graph $G=(V,E)$, where vertices denote patient-donor pairs and edges indicate mutually compatible exchanges.

 As opposed to Kern et al.~\cite{kern2019generalized}, in this paper we focus on the non-transferable utility (NTU) setting, which is a natural fit for kidney exchange: the utility (number of patients matched) cannot be transferred among players. This is also in line with how the literature for the US market models the utilities of the participants~\cite{ashlagi2015mix,ashlagi2011individual,ashlagi2014free,roth2005pairwise,roth2005transplant}. We investigate the existence and computational complexity of core solutions in NTU partitioned matching games, particularly exploring two important stability notions: the strong core, where no coalition can weakly block the solution (meaning only one participant has to improve), and the weak core, where no coalition can strongly block (i.e. all participants must improve) it.

Beyond the international setting, our framework also applies to national KEPs like the US setting, where the players may represent hospitals rather than countries~\cite{sonmez2013market}. In such cases, each institution typically controls only a small number of patient-donor pairs. This motivates the study of partitioned matching games where each player controls a bounded number of vertices.

\subsection{Motivation: NTU matching game with couples}

The NTU Partitioned Matching Game has many remarkable properties also from the point of view of algorithmic game theory and computational complexity. It turns out that the most interesting case from this point of view is when each player controls two vertices of the graph, which we call the \emph{NTU matching game with couples} (we remark that the case where each partition class has size \emph{at most two} can be easily reduced to the case where each partition class has size \emph{exactly two}, by introducing dummy vertices). Although this case is not very relevant to the kidney exchange problem, we believe that it is interesting on its own right, and its analysis requires sophisticated arguments based on matching theory. To have an intuitive understanding of this game, we can imagine some couples who come together to play some sport, for example squash, with two-player matches. Each individual has binary preferences, i.e. there are people they like to play with, and there are those they prefer not to play with (we assume that these preferences are mutual, and of course it is possible that a couple prefers not to play with each other). The task is to organize the matches in a way that is acceptable to all couples, where the best option for a couple is that both of them play with a preferred partner, and the second best option is that one of them plays with a preferred partner. In this setting, the core is a natural and intuitive solution concept, because it means that no subset of couples has the incentive to leave the group and organize their matches among themselves.

As we will see, the weak core is always non-empty in this case, but the non-emptiness of the strong core is an intriguing question that is intimately related to the structure of paths that alternate between edges of the graph and ``edges'' corresponding to couples. Furthermore, although the previous statement suggests that the structure of the weak core is simpler than that of the strong core, this is in fact not the case: it turns out that optimizing a linear function on the weak core is NP-hard, while the same problem on the strong core is solvable in polynomial time.

Our results on this problem fit into a long sequence of results on matching problems with couples, where the introduction of joint preferences of couples leads to intriguing complexity questions \cite{ronn1990np,aldershof1996stable,klaus2005stable,marx2011stable,mcbride2013hospitals,biro2016matching,nguyen2018near,csaji2023couples}.

\iffalse
From the theoretical viewpoint, a particularly interesting special case arises when each player controls exactly two vertices -- modeled as a couple. We refer to this subclass as the NTU matching game with couples, and show how it models a variety of matching situations beyond kidney exchange. For example, consider couples organizing squash matches: each individual has a list of preferred opponents, and the best outcome for a couple is when both partners get to play. Alternatively, one might consider a research exchange program in which couples represent student-supervisor pairs, and matching depends on mutual academic interests. In both scenarios, the key question is whether a stable matching exists such that no subset of couples can coordinate a strictly better outcome among themselves.
\fi

\subsection{Our results}

The motivations mentioned above justify the study of the computational complexity both for instances where the size of each partition class is bounded, and for instances where the number of partition classes (i.e., players) is bounded.
Our main results are as follows:
\begin{itemize}

    \item For the NTU matching game with couples, we show that the weak core is always nonempty and a weak core solution can be found in polynomial time. As the technically most involved result of the paper, we construct a polynomial-time algorithm for deciding if the strong core is empty, and for finding a strong core solution, if it exists. The proof uses an analysis of the structure of paths and cycles that are alternating between edges of the graph and ``edges'' corresponding to the couples. Our results also imply that we can efficiently optimize on the strong core, i.e., for any given edge costs, we can find a minimum cost matching in the strong core in polynomial time. This is remarkable since optimizing on the weak core turns out to be an NP-hard problem.

    \item We show that if the number of players is bounded by a constant, then the problems of checking membership in the weak and strong cores, and of finding minimum cost elements in the weak and strong cores are all polynomial-time solvable. 

    \item In contrast to the above results, we show that problems where the size of each partition class is bounded by a constant are hard even for fairly small constants. In particular, deciding whether a matching is in the weak or strong core is coNP-complete even when each partition class has size at most 3. Similarly, deciding the weak or strong core's emptiness also becomes NP-hard, even for small sizes.
\end{itemize}

%\textbf{Ez új, de a MATCH-Up bírálatban hiányoltak egy ilyen részt}

Our results provide useful new tools for designing mechanisms for IKEPs. In particular, as most IKEPs consist of only a few countries, our results imply that we can efficiently find strong and weak core solutions for the NTU version of the partitioned matching game in those cases. Such solutions, if they exist, give desirable outcomes for exchange programs, as incentives to quit the program can be avoided. Also, as membership of a matching in the strong or weak core depends only on the number of vertices it covers from each player, a nonempty core always contains a maximum size matching (this follows from the fact that any set of vertices coverable by a matching is also coverable by a maximum size matching). Therefore, our concepts are compatible with the most important restriction used in KEPs, namely that the number of transplants should be maximized. Furthermore, as a simulation study of Colette~\cite{collette2023weak} -- using the state-of-the-art KEP generator by Saidman et al.~\cite{saidman2006increasing} -- recently showed, the weak core of realistic international kidney exchange markets always admits a weak core exchange. Hence, even though existence is not guaranteed in theory, instances with an empty weak core arise only in special constructions, but not in real-world markets. Furthermore, this also highlights that even though most arising problems are NP-hard in theory, it is feasible to compute core solutions in practice with the help of integer programming techniques. 

\iffalse
Kidney exchange programs can also include altruistic donors -- i.e. donors without a related patient in need of a kidney -- who can start a chain of transplantations. 
While our model does not deal with altruistic donors, this is in line with most game-theoretical literature on kidney exchange, both in Europe~\cite{benedek2023computing,benedek2025partitioned,benedek2021computing,kern2019generalized} and in the US \cite{ashlagi2015mix,ashlagi2014free,roth2005transplant}. Furthermore, with the help of dummy patients that are compatible with everyone, we can easily model the presence of altruistic donors even in a model which allows cycles of exchanges only. 
\fi

In addition to the applicability of our results to kidney exchange problems, our algorithm for the problem with couples offers a new class of non-bipartite matching games where core solutions can be efficiently found even in the presence of couples having joint preferences.

%The mechanism then computes a maximum size matching that approximates this target allocation for each country. At the end of the round, the scores are updated according to whether the number of matched patients of a country is above or below the target. 

\paragraph{Structure of the paper}
The rest of this section describes related work and introduces the necessary definitions and notation. In subsequent sections, we study the computational complexity of membership and non-emptiness of the weak and strong cores, under various assumptions on the players.

In Section \ref{sec:size2}, we show how membership in the weak and strong cores in the NTU matching game with couples can be characterized using alternating paths. We then prove that the weak core is always non-empty in this case; however, optimizing on the weak core is NP-hard. The strong core can be empty, but we are able to give a structural characterization and a polynomial-time algorithm for finding a minimum cost strong core solution if it exists. 

In Section \ref{sec:constno}, we show that all problems studied are polynomial-time solvable if the number of players is constant, even if the sizes of the partition classes can be arbitrarily large. 
Section \ref{sec:hard} presents hardness results for the case where the partition classes have bounded size. We show that all problems are hard for some small constant bound, and for most problems we can prove coNP-hardness even for partition classes of size at most 3.

\subsection{Related work}

\paragraph{Transferable Utility (TU) Matching Games.}
Matching games with transferable utilities are known as assignment games. A foundational result by Shapley and Shubik~\cite{shapley1971assignment} shows that the core of any assignment game is non-empty. This result was extended to many-to-many settings by Sotomayor~\cite{sotomayor1992multiple}, who proved core non-emptiness for any capacity function $b$.

Computational aspects of core-related problems have also been studied. Biró et al.~\cite{biro2018stable} showed that deciding whether an allocation lies in the core is coNP-complete even for uniform 3-assignment games. For $b\le 2$, however, they proved that both core verification and finding a core solution are solvable in polynomial time.

The partitioned matching game was introduced by Kern et al. \cite{kern2019generalized} to model international kidney exchange programs (IKEPs), where players control subsets of vertices rather than individual ones. Later works~\cite{benedek2021computing,benedek2025partitioned} examined the computational complexity of these games and showed that the core can be empty. They also established NP-hardness and coNP-hardness for core non-emptiness and verification, even when players control at most three vertices. A close relationship was established with many-to-many assignment games via two-way reductions.

A variant of these games involving unbounded exchanges, called partitioned permutation games, was studied by Benedek et al.~\cite{csajiunbounded}. They demonstrated that such games always admit a core solution, which can be found in polynomial time. Nevertheless, verifying whether a given solution belongs to the core remains coNP-complete.

\paragraph{Non-Transferable Utility (NTU) Matching Games.}
In the NTU setting, Scarf’s classical result~\cite{scarf1967core} establishes the non-emptiness of the (weak) core in a broad class of games, called balanced games. Kaneko~\cite{kaneko1982central} applied Scarf's Lemma to the NTU version of the Shapley-Shubik assignment game and showed that the weak core is always non-empty, even in general settings. Biró and Fleiner~\cite{biro2016NTU} further used Scarf’s Lemma to establish the existence of fractional weak core solutions in hypergraph matching games.

Stability is a concept closely related to the core. In many-to-one matching markets, Roth~\cite{roth1984stability} showed that under strict preferences, the strong core coincides with stable matchings. Gale and Shapley~\cite{gale1962college} proved the existence of stable matchings in bipartite graphs, while Irving~\cite{irving1985efficient} gave a polynomial algorithm for non-bipartite settings with strict preferences, where existence is no longer guaranteed.

However, in many-to-many markets, the connection between core and stability weakens. Blair~\cite{blair1988lattice} and Sotomayor~\cite{sotomayor1999three} demonstrated that the strong core may be empty and does not coincide with stable outcomes. Recently, Biró and Csáji~\cite{biro2024strong} showed that deciding strong core non-emptiness is NP-hard even in restricted many-to-many markets, and that verifying core membership is coNP-hard.

The NTU version of the partitioned matching game was first explored in our earlier technical report~\cite{egres-19-12}, which established basic complexity results. Independently, Collette's Master's thesis~\cite{collette2023weak} also studied the weak core in NTU partitioned matching games. He proved that the weak core is always non-empty for two players, and provided a no-instance for the core with 3 players - with 9 vertices each, compared to our no-instance with 7 vertices per player. Then, he proposed a (not polynomial) cutting-plane-based algorithm to compute weak core outcomes. Using the Saidman generator~\cite{saidman2006increasing}, he showed empirically that weak core matchings always existed in simulations, highlighting the practical viability of this concept. He also compared weak core outcomes with those of rejection-proof mechanisms proposed by Blom et al.~\cite{blom2022rejection}.

\paragraph{Incentives in kidney exchange programs.}

\new{The literature on incentive issues in kidney exchange programs can be broadly classified into three main streams, depending on whether they adopt a cooperative, non-cooperative, or mechanism-design perspective. We summarize these below and highlight their relationship to our approach.}

\smallskip
\new{\emph{Cooperative approaches.}  
In the European setting, a common approach is to treat countries as cooperating players and to focus on fairness over time across multiple exchange rounds \cite{hajaj2015strategy,kern2019generalized,klimentova2021fairness}. In these systems, each country is assigned a ``target'' number of transplants in each round based on a fairness criterion. Because the actual matches in a round may deviate from these targets due to medical and logistical constraints, discrepancies are tracked as credits (positive or negative) and carried forward to future rounds. Over time, this mechanism balances outcomes across rounds.  
Kern et al.~\cite{kern2019generalized} proposed using solution concepts from transferable utility (TU) cooperative game theory to determine these target allocations, such as the benefit value, contribution value, Shapley value, or Banzhaf value. A key advantage of this approach is that it is \emph{compatible with maximum-size matching}: in each round, the actual matching is chosen to maximize the total number of transplants, and fairness is addressed through the credit adjustments across rounds.}

\smallskip
\new{
\emph{Non-cooperative approaches.}  
In the US setting, hospitals are typically modeled as self-interested players in a non-cooperative game. Much of the literature has focused on strategy-proofness and individual rationality (IR), where IR in NTU terms corresponds to the absence of a single blocking player \cite{ashlagi2011individual,ashlagi2012new,ashlagi2014free}.  
A major challenge in this setting is that strong incentive guarantees come at a cost to efficiency. Ashlagi and Roth~\cite{ashlagi2014free} proved that no mechanism can be both strategy-proof and individually rational while guaranteeing more than half of the maximum possible number of transplants, even when only 2-way cycles are allowed. Ashlagi et al.~\cite{ashlagi2015mix} proposed a randomized mechanism that achieves at least half the maximum exchange size and proved that no mechanism can exceed 7/8 of the optimal size under these constraints.  
Carvalho et al.~\cite{carvalho2017nash} studied the two-hospital case and characterized Nash equilibria, while Carvalho and Lodi~\cite{carvalho2023Nash} extended these results to social welfare equilibria for an arbitrary number of players, also assuming a bound of 2 on the exchange length, as in our model. These equilibria are compatible with maximum-size solutions, but the authors do not address coalition deviations, which are central to our setting.  
Gourvès et al.~\cite{gourves2008cooperation} showed that deciding individual rationality in weighted bipartite matchings is NP-hard, illustrating the computational complexity of incentive-related questions in this domain.}

\smallskip
\new{
\emph{Mechanism design for clearinghouses.}  
Roth et al.~\cite{roth2005pairwise} introduced a priority-based mechanism for pairwise exchanges that is strategy-proof and ensures that revealing all available pairs is a dominant strategy. Sönmez and Ünver~\cite{sonmez2013market} generalized this approach, showing that there exist mechanisms that are simultaneously strategy-proof and Pareto-optimal \emph{only for pairwise exchanges}. However, they also proved that no such mechanism exists once longer cycles or chains are allowed, which fundamentally limits this line of work.  
Blom et al.~\cite{blom2022rejection} proposed a rejection-proof mechanism, meaning that it is a weakly dominant strategy for participants to accept its outcome given that others also accept.}

\smallskip
\new{
\emph{Our contribution and relation to the literature.}  
Our work combines ideas from cooperative and non-cooperative perspectives. We model IKEPs as cooperative games with non-transferable utility, where players (countries or hospitals) seek to maximize the number of their patients who receive transplants. We focus on two stability notions: the strong core and the weak core.  
Unlike many non-cooperative mechanisms, these notions do not require a loss in efficiency: we show that whenever the strong or weak core is nonempty, there always exists a core solution that achieves the maximum possible number of transplants. Thus, our framework provides a path to designing systems that incentivize participation while fully preserving efficiency. This distinguishes our approach from much of the prior literature, where incentive guarantees often come at the expense of the total number of transplants.}

\new{We note that kidney exchange programs can also include altruistic (nondirected) donors -- i.e. donors without a related patient in need of a kidney -- who can start a chain of transplantations \cite{roth2006utilizing,ashlagi2012new}. 
While our model does not deal with altruistic donors, this is in line with most game-theoretical literature on kidney exchange, both in Europe~\cite{benedek2023computing,benedek2025partitioned,benedek2021computing,kern2019generalized} and in the US \cite{ashlagi2015mix,ashlagi2014free,roth2005transplant}.}

\paragraph{Empirical studies.}
Biró et al.~\cite{biro2019building} provided a comprehensive overview of European kidney exchange programs. Many countries, such as Norway, Denmark, Sweden, and France, only allow pairwise exchanges, while others like the UK permit up to 3-way exchanges. Druzsin et al.~\cite{druzsin2024performance} conducted simulation studies demonstrating that allowing full cooperation -- rather than restricting international exchanges to hard-to-match patients -- could lead to significantly more transplants.

Mechanisms based on multiround credit systems have also been studied~\cite{hajaj2015strategy,kern2019generalized,klimentova2021fairness}. In these approaches, cooperative game theory concepts are used to determine target allocations, and matchings in each round are selected to approximate these targets. Benedek et al.~\cite{benedek2021computing} conducted large-scale simulations, while~\cite{benedek2025partitioned} analyzed the theoretical complexity of these problems.

\subsection{Preliminaries}
For a positive integer $k$, we use the notation $[k]=\{1,2,\dots,k\}$. An instance $I$ of the \emph{NTU Partitioned Matching Game} is defined by a tuple $(G,\partition)$, where $G = (V;E)$ is a graph and $\partition=(V_1, V_2, \dots, V_m)$ is a partition of $V$ into $m$ disjoint subsets, where $m$ is the number of players. Each class of this partition corresponds to a player; we say that the vertices in $V_i$ belong to player $i\in [m]$ (sometimes we identify the player and the set and say that player $i$ has size $|V_i|$). If each $V_i$ has size 2, then the problem is called the \emph{NTU matching game with couples}.

%We denote the set of players by $N$. 
Given a matching $M$ in $G$, the utility of $M$ for player $i$ is defined by $u_i(M)=|V(M)\cap V_i|$, 
where $V(M)$ denotes the vertex set of $M$; i.e., the utility is the number of vertices of $V_i$ covered by $M$.

A coalition $\mathcal{P}=\{ i_1, i_2,\dots,i_k \}$ of players is \emph{strongly
blocking} for the matching $M$ if there exists a matching $M'$ in the
induced subgraph $G[V_{i_1} \cup V_{i_2}\cup \dots \cup V_{i_k}]$ such that $u_{i_j}(M')>u_{i_j}(M)$ for every $j \in [k]$. Similarly, a coalition $\mathcal{P}=\{ i_1, i_2,\dots, i_k\}$ of players is \emph{weakly blocking} for a matching $M$ if
there exists a matching $M'$ in the
induced subgraph $G[V_{i_1} \cup V_{i_2}\cup \dots \cup V_{i_k}]$ such that $u_{i_j}(M') \geq u_{i_j}(M)$ for every $j \in [k]$, and $u_{i_j}(M') > u_{i_j}(M)$ for at least one $j \in [k]$.

%The Partitioned Matching Game, in a transferable utility (TU) version was introduced by Biró et al.~\cite{kern2019generalized} motivated by the international kidney exchange problem. In that setting, the vertices correspond to patient-donor pairs, and edges represent the possible pairwise exchanges. The partition $V_1,\dots,V_m$ can be thought of as a partitioning of the patients according to their countries, but we can also think of hospitals instead of countries, with each hospital being interested in successful transplants for their own patients.
%The problem is studied from the point of view of cooperative games: are there matchings that are acceptable for all possible coalitions of players?

In the NTU setting, a matching $M$ is in the \emph{weak core} if there is no strongly blocking coalition for it, and it is in the \emph{strong core} if there is no weakly blocking coalition. It is easy to see that if $M$ is in the strong core, then it is also in the weak core. Deciding whether a given matching $M$ is in the strong or weak core is a problem in coNP, since a blocking coalition, together with the matching $M'$ that certifies the blocking, is a suitable witness.

Given an instance $I=(G,\partition)$ of the NTU partitioned matching game with partition classes of size at most $\ell$, we denote the problems of deciding the existence of a weak core or strong core matching by \weakex-$\ell$ and \strongex-$\ell$, respectively. Similarly, we denote the problem of deciding if a given matching $M$ in $I$ is in the weak or strong core by \weakverif-$\ell$ and \strongverif-$\ell$, respectively. The $\ell = 2$ case corresponds to the NTU matching game with couples.

The following lemma is a central tool for finding matchings in the weak or strong cores, and for verifying if a given matching is in the weak or strong core. 
\begin{lemma}\label{lem:lowerbound}
Given a graph $G = (V;E)$, a partition $U_1 \cup U_2 \cup \dots \cup U_l$ of $V$, and a vector $q \in \Zset_+^m$ \new{with $q_i\le |U_i|$}, it can be decided in $\mathcal{O}(|V|^{2.5})$ time if there is a matching $M$ such that $|U_i\cap V(M)| \geq q_i$ for every $i \in [l]$. Furthermore, given a cost function $c:E\to \mathbb{R}$, we can find a minimum cost such matching in $\mathcal{O}(|V|^4)$ time, if one exists.
\end{lemma}

\begin{proof}
Create a graph $\hat{G}=(\hat{V},\hat{E})$ as follows. For each $U_i$ ($i\in [l])$, we add $k_i=|U_i|-q_i$ dummy vertices $d_i^1,\dots, d_i^{k_i}$. All of these vertices are connected to the vertices in $U_i$ and each other. Furthermore, if the number of vertices so far is odd, then we add a further dummy vertex $d^*$, connected to all dummy vertices. Then, it is easy to see that there exists a matching $M$ in $G$ such that $|V(M)\cap U_i|\ge q_i$ $\forall i\in [l]$, if and only if there exists a perfect matching $\hat{M}$ in $\hat{G}$. As the creation of $\hat{G}$ can be done in $\mathcal{O}(|V|^2)$ time, and a perfect matching in $\hat{G}$ can be found in $\mathcal{O}(|V|^{0.5}|\hat{E}|)$-time, we get the $\mathcal{O}(|V|^{2.5})$ runtime as $|\hat{E}|\le (2|V|)^2$.

To find a minimum cost matching $M$ with $|V(M) \cap U_i| \geq q_i$ $\forall i\in [l]$, if one exists, given a cost function $c:E\to \mathbb{R}$, let $R:=\max \{ |c(e)|\mid e\in E\} + 1$. It suffices to find a minimum cost perfect matching in $\hat{G}$ with the cost function $c'(e):=c(e)$, if $e\in E$ and $c'(e)=0$ for $e\in \hat{E}\setminus E$. 

A minimum cost matching can be computed with Edmonds' Blossom algorithm in $\mathcal{O}(|V|^2|\hat{E}|)$ time, so the claimed runtime also follows.
\end{proof}

\section{The NTU matching game with couples}\label{sec:size2}

In this section, we study the case where every partition class has size at most 2. 
%In contrast with the previous hardness results,
We show that in this case, deciding membership and deciding existence are efficiently solvable for both the weak and strong cores, although deciding the emptiness of the strong core requires a complicated approach relying on many structural results. Somewhat surprisingly, the computational
difficulty of optimization is the other way around: it is NP-hard for the weak core, while polynomial-time solvable for the strong core.

\paragraph{The set $M_0$ of player edges.} First of all, we may assume without loss of generality that every partition class has size exactly two; otherwise we add an isolated vertex and make it belong to that player. 
With this assumption, the partition classes can be considered as graph edges, which we call \emph{player edges}. %The player edge for a player $A$ will be denoted by $e_A$. 
The perfect matching on vertex set $V$ defined by the player edges is denoted by $M_0$. For a player $A$, we will denote the corresponding player edge by $\pedge_A$. These edges are not part of the original graph $G=(V,E)$, although $E$ can also contain an edge $e=aa'$ between the two vertices of a player $A=\{ a,a'\}$, in which case $\pedge_A$ will be a parallel edge to $e$. See Figure~\ref{fig:pedges} for an illustration.

\begin{figure}
    \centering
    \includegraphics[width=0.4\linewidth]{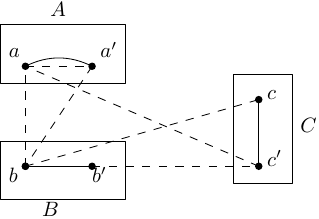}
    \caption{An example of the graph supplemented with the player edges. The dashed edges denote $E$ and the solid edges denote the player edges $M_0$.}
    \label{fig:pedges}
\end{figure}

We say that a path or cycle in $E \cup M_0$ is \emph{$M_0$-alternating} if it alternates between edges in $M_0$ and edges in $E$. Note that if both $E$ and $M_0$ contain edges between vertices $u$ and $v$, then these two parallel edges form an $M_0$-alternating cycle. 

For an $M_0$-alternating path or cycle $P$, we can consider the coalition formed by $\mathcal{P}=\{ i\in [m]\mid \pedge_i\in P\}$, and the matching $P\cap E$ on the vertices of the coalition. Let $M\subseteq E$ be an arbitrary matching; we say that $P$ is a (weakly/strongly) \emph{blocking path or cycle for the matching $M$}  if $\mathcal{P}$ is a weakly/strongly blocking coalition for $M$, as witnessed by the matching $P\cap E$. 

\subsection{Weak core}

As we shall see, the problems concerning the weak core are considerably simpler than those concerning the strong core, except for the optimization problem. 
We begin with a characterization result for the weak core. 

\begin{lemma}\label{2weak}
A matching $M$ is in the weak core if and only if
\begin{enumerate}[a)]
  \item there is no $M_0$-alternating cycle such that every player edge in it has at most one vertex covered by $M$,
  \item there is no $M_0$-alternating path such that the player edges at the two ends of the path have 0 vertices covered by $M$, and the player edges in the middle have one vertex covered by $M$.
\end{enumerate}
\end{lemma}
\begin{proof}
If $M$ is in the weak core, then conditions a) and b) hold, since otherwise the $M_0$-alternating path or cycle would block $M$. Suppose a) and b) hold, but $M$ is not in the weak core. There is a coalition of players with a matching $M'$ such that every player in the coalition is strictly better off with $M'$, therefore the players who have both of their vertices covered by $M$ cannot belong to the coalition -- we can delete the vertices of these players from the graph, and consider only the remaining graph. Take the union of the remaining player edges and $M'$. This consists of $M_0$-alternating cycles and paths. If there is an $M_0$-alternating cycle in the union, then it contradicts a). If there is an $M_0$-alternating path, then the player edges at the two ends of the path have one vertex covered by $M'$, therefore they have 0 vertices covered by $M$, which contradicts b).
\end{proof}
\begin{theorem} \label{thm:2weakverif}
 \weakverif-2 is solvable in $\mathcal{O}(|V|^{4.5})$ time. %That is, we can decide in $\mathcal{O}(|V|^{2.5}|E|)$ time if a given matching $M$ in an instance of the NTU matching game with couples is in the weak core.

\end{theorem}
\begin{proof}
We show that we can check the conditions of Lemma \ref{2weak} in $\mathcal{O}(|V|^{4.5})$ time (utilizing that a maximum size matching can be computed in $\mathcal{O}(|V|^{2.5})$ time~\cite{blum1990new}). Let $M'_0 \subseteq M_0$ be the set of player edges that have at most one vertex in $V(M)$. Let $G'$ denote the graph obtained by restricting $G$ to $V(M'_0)$, and adding $M'_0$ to the edge set.

We can check if a given player edge $\pedge_A$ is in an $M'_0$-alternating cycle, since it is equivalent to checking if, after the deletion of this edge from $G'$, the remaining graph has a perfect matching. Indeed, the remaining player edges form a matching that covers all but two vertices of $G'$ (an \emph{almost perfect matching}), and there is an augmenting $M'_0$-alternating path between the two vertices of $A$ if and only if $\pedge_A$ is in an $M'_0$-alternating cycle in $G'$. This shows that condition a) can be checked in $\mathcal{O}(|V|^{3.5})$ time, by repeating the above for all player edges in $M'_0$.

We can also check whether there is an $M'_0$-alternating path between two given player edges $\pedge_A$ and $\pedge_B$ that have 0 vertices covered by $M$ and that do not belong to an $M'_0$-alternating cycle. To do this, we remove $\pedge_A$ and $\pedge_B$ from $G'$, and check if the remaining graph has an almost perfect matching. If such a matching exists, then there is an augmenting $M'_0$-alternating path between two of the four vertices of $A\cup B$, and this path cannot connect the two vertices of $A$ or the two vertices of $B$, because then $\pedge_A$ or $\pedge_B$ would be in an $M'_0$-alternating cycle in $G'$. Thus, condition b) is violated. Conversely, if there is no almost perfect matching after the deletion of $\pedge_A$ and $\pedge_B$, then there is no $M'_0$-alternating path where the two ends are $\pedge_A$ and $\pedge_B$. Therefore, condition b) can be checked by iterating through all possible choices of $A$ and $B$ in time $\mathcal{O}(|V|^{4.5})$.
\end{proof}

The next natural problem would be to decide whether the weak core is empty. However, it turns out that the weak core of the NTU matching game with couples cannot be empty.

\begin{theorem}
The weak core is never empty, hence \weakex-2 can be solved in $\mathcal{O}(1)$ time. Furthermore, a matching in the weak core can be found in $\mathcal{O}(|V|^{3.5})$ time.
\end{theorem}
\begin{proof}
We can construct a matching $M^*$ in the weak core the following way. Let $M_0$ be the perfect matching formed by the player edges as before. Iterating through the player edges, we check if there is an $M_0$-alternating cycle in the current graph containing that edge. If there is, then let the edges of $E$ in the cycle belong to $M$, and delete the vertices of the cycle from the graph. 
Repeat this with the remaining graph, until there are no $M_0$-alternating cycles. Note that each player edge will be checked at most once, as deleting edges cannot create an $M_0$-alternating cycle. Hence, this can be done in $\mathcal{O}(|V|^{3.5})$ time with maximum matching algorithms.

Let $M$ denote the resulting matching, and let $M^*$ be a maximum size matching in $G$ that covers every vertex of $V(M)$. This can be computed in $\mathcal{O}(|V|^{2.5})$ time. We claim that $M^*$ is in the weak core, since the conditions of Lemma \ref{2weak} are met. Condition a) clearly holds, since the construction deleted players that were covered twice by $M$ (and hence by $M^*$), and the remaining graph did not contain an $M_0$-alternating cycle.

Suppose for contradiction that condition b) does not hold, i.e.\ there is an $M_0$-alternating path $P$ such that the player edges at the two ends of the path have 0 vertices covered by $M^*$, and the player edges in the middle have one vertex covered by $M^*$. Let $N^*$ be the set of edges of $M^*$ that contain a vertex from the path $P$. We have $|N^*| \leq |P \cap M_0|-2=|P\cap E|-1$. Therefore, $(M^*\setminus N^*) \cup (P \cap E)$ is a larger matching than $M^*$, contradicting the choice of $M^*$ as a maximum-size matching.
\end{proof}

Since the weak core is always non-empty, it is natural to consider the problem of finding a minimum-cost matching in the weak core according to some linear cost function. However, this problem turns out to be hard even for players of size one, as the following theorem shows (formally, we can add an isolated vertex to each player in the theorem and the proof below if we want to get an instance with couples).

\begin{theorem}\label{thm:1opthard}
The following problem is NP-complete even for players of size 1: given an instance $(G,\partition)$ of \weakex-2, a cost vector $c \in \mathbb{Z}^E_+$, and a positive integer $k$, find a matching with cost at most $k$ in the weak core. 
\end{theorem}
\begin{proof}
    Observe that if each player has size one, then a matching $M$ is in the weak core if and only if $V(M)$ is a vertex cover of the graph. This enables us to reduce {\sc Vertex Cover} to our problem as follows. Let $(G=(V,E), k)$ be an instance of {\sc Vertex Cover}, where we want to decide if there is a vertex cover of $G$ of size at most $k$. We construct a graph $G'$ by adding new vertices $v_1,\dots,v_k$, and adding edges from each $v_i$ to every vertex in $V$. The original edges have cost $k+1$, while the new edges have cost 1. 
    
    We claim that $G$ has a vertex cover of size at most $k$ if and only if the weak core of $G'$, with players of size 1, contains a matching of cost at most $k$. Indeed, a matching of $G'$ with cost at most $k$ cannot contain any original edges and at most $k$ new edges, and such a matching is in the weak core if and only if it contains exactly $k$ new edges and their endpoints in $V$ form a vertex cover of $G$.
\end{proof}

\subsection{Strong core}

 Our first result on the strong core is a characterization similar to Lemma \ref{2weak}.

\begin{lemma}\label{2strong}
 A matching $M$ is in the strong core if and only if
\begin{enumerate}[a)]
  \item every player edge that is in an $M_0$-alternating cycle has both its vertices covered by $M$,
  \item there is no $M_0$-alternating path from a player edge that has none of its vertices covered by $M$ to a player edge that has at most one vertex covered by $M$,
  \item there is no $M_0$-alternating path such that at least three player edges in the path have at most one vertex covered by $M$.
\end{enumerate}
\end{lemma}
\begin{proof}
It is easy to see that if $M$ is in the strong core, then these conditions hold, since otherwise the $M_0$-alternating path or cycle would block $M$. Suppose that the conditions hold, but the matching $M$ is not in the strong core. Then there is a weakly blocking coalition of players, together with a matching $M'$, such that there is one player that is better off with $M'$ than $M$, and the others are not worse off. Take the union of $M'$ and $M_0$. This is a disjoint union of $M_0$-alternating cycles and paths. If the player that is better off by $M'$ has 0 vertices covered by $M$, then it has at least one vertex covered by $M'$, so it is either in an $M_0$-alternating cycle, which contradicts condition a), or it is in an $M_0$-alternating path. In the latter case, the two player edges at the ends of the path are players who have one vertex covered by $M'$, so they have at most one vertex covered by $M$, which contradicts condition b). If the player that is better off by $M'$ has one vertex covered by $M$, then it has two vertices covered by $M'$, so it is either in an $M_0$-alternating cycle, which contradicts condition a), or it is in an $M_0$-alternating path (and it is not at the end of the path), but the two player edges at the ends of the path are players that have one vertex covered by $M'$, so they have at most one vertex covered by $M$, which contradicts condition b) or c).
\end{proof}

\begin{theorem}\label{2score-verif}
\strongverif-2 can be solved in $\mathcal{O}(|V|^{5.5})$ time. %That is, in the NTU matching game with couples, we can decide in $\mathcal{O}(|V|^{3.5}|E|)$ time if a given matching $M$ is in the strong core.
\end{theorem}

\begin{proof}
We need to show that we can check the conditions given by Lemma \ref{2strong} in polynomial time. The techniques used here are similar to the proof of Theorem \ref{thm:2weakverif}, so we skip some details. We can check if a given player edge is in an $M_0$-alternating cycle, since it is equivalent to checking if, after the deletion of this edge, the remaining graph (the original graph $G$ together with the remaining player edges) has a perfect matching. We can do this for every player edge with at most one vertex in $V(M)$. 

We can check whether there is an $M_0$-alternating path between two given player edges, knowing that these do not belong to an $M_0$-alternating cycle, since this is equivalent to checking if by deleting these two edges the remaining graph has an almost perfect matching. We can do this for any two player edges such that one of them has 0 vertices covered in $V(M)$, and the other one has at most one vertex in $V(M)$.

For three given player edges such that none of them belong to an $M_0$-alternating cycle, we can check if there is an alternating path that contains all three of them: this is equivalent to checking if, by deleting all three of these edges, the remaining graph has an almost perfect matching. We can check this for any three player edges that have one vertex in $V(M)$. This leads to a runtime of $\mathcal{O}(|V|^{5.5})$.
\end{proof}

\subsubsection{Deciding the existence of a strong core matching} \label{sec:scexistence}

Our aim in this subsection is to prove that it can be decided in polynomial time if the strong core is nonempty, and if so, a matching in the strong core can be found in polynomial time. These results require a careful analysis of the structure of $M_0$-alternating paths, so the proof will feature several structural lemmas. We start by introducing some notation. 

Let $\SC $ denote the set of matchings in the strong core.
We introduce some additional notation and definitions related to the structure of $M_0$-alternating paths.

\paragraph{$M_0$-alternating paths.} For three players $A$, $B$ and $C$, an \emph{$M_0$-alternating $A-B-C$ path} is an $M_0$-alternating path that starts with the player edge $\pedge_A$, ends with the player edge $\pedge_C$, and also contains the player edge $\pedge_B$.
    
    \paragraph{$M_0$-alternating $\delta$-paths.} An \emph{$M_0$-alternating $\delta$-path} (or \emph{$\delta$-path} for short) $H$ is an odd cycle and an $M_0$-alternating path with a common vertex $v$, where the $M_0$-edge adjacent to $v$ is on the path, and the cycle is $M_0$-alternating apart from $v$. The cycle of $H$ is denoted by $\Hcycle (H)$ and the path of $H$ is denoted by $\Hpath (H)$.
    An $(A,B;C)$-$\delta$-path for some players $A,B,C$ is an $M_0$-alternating $\delta$-path where $\pedge_A,\pedge_B$ are in $\Hcycle (H)$ and $\pedge_C$ is the end of the path $\Hpath(H)$. The vertex of $C$ at the end of the path is called the \emph{tip} and the unique common vertex of $\Hpath (H)$ and $\Hpath (H)$ is called the \emph{connector}. Note that the roles of $A$ and $B$ are exchangeable in the definition. See Figure \ref{fig:deltapath} for an illustration. 
\begin{figure}[h]
    \centering
    \includegraphics[width=0.6\linewidth]{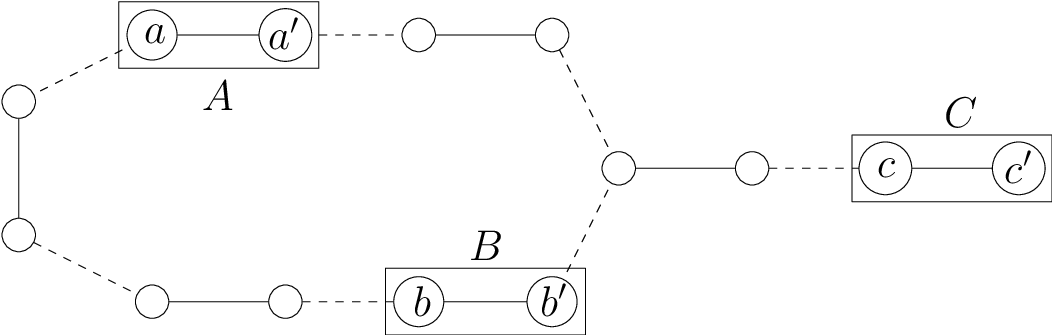}
    \caption{An $(A,B;C)$-$\delta$-path. Solid edges denote the player edges, i.e. edges of $M_0$, while dashed edges are in $E$. }
    \label{fig:deltapath}
\end{figure}

\paragraph{The sets of players $\nocyc,\nocyc_0,\nocyc^*,\nopath,\nopath^*
$.}
We proceed to define some additional subsets and relations of players, that will be used extensively in the proofs.
\begin{itemize}
    \item Let $\nocyc$ denote the set of players for whom there is no \altcyc\ through the corresponding player edge.
    \item Let $\nocyc_0\subseteq \nocyc$ denote the set of players $A \in \nocyc$ from which there is no $M_0$-alternating path to any other player in $\nocyc$.
    \item We denote by $\nocyc^*$ the set of players $A$ for whom $\exists M \in \SC$: $|V(M)\cap A| \leq 1$. 
  
\end{itemize}

Of course, we cannot determine $\nocyc^*$ without deciding whether the strong core is empty, so we will not use $\nocyc^*$ in the algorithm, only in the proof of correctness. By Lemma \ref{2strong} point a), we have $\nocyc^*\subseteq \nocyc$.

\begin{itemize}
      \item We define a subset $\nopath \subseteq \nocyc$ the following way. A player $B\in \nocyc$ is in $\nopath$, if for any $M_0$-alternating $A-B-C$ path for any $A \in \nocyc$ and $C \in \nocyc$, there exists an $(A,B;C)$-$\delta$-path or a $(C,B;A)$-$\delta$-path.
%at least one of the three possibilities in Lemma \ref{lem:cases} holds.

\item For $A,B \in \nopath$, we say that $(A,B)$ is a \emph{pair} if there exists an $(A,B;C)$-$\delta$-path for some $C\in \nocyc$.

\item Define $\nopath^* \subseteq \nopath$ the following way: a player $A \in \nopath$ is in $\nopath^*$ if for every $B,C \in \nopath$ such that $(A,B)$ and $(A,C)$ are pairs, $(B,C)$ is also a pair.

\end{itemize}

\paragraph{The graph $G^*$.} We introduce an auxiliary graph $G^*$ whose vertex set is $\nopath^*$ and $(A,B)$ is an edge if and only if it is a pair. We will show that this graph can be constructed in polynomial time. By the definition of $\nopath^*$, every connected component of $G^*$ is a clique. 
 The members of $\nocyc_0$ are isolated vertices in $G^*$ (but there may be other isolated vertices). Let $\mathcal{Q}$ denote the set of connected components of $G^*$. For a connected component $Q \in \mathcal{Q}$, let $V_Q \subseteq V$ denote the set of vertices of all players in $Q$.

We will illustrate these notions in an example below. But first, let us state our main theorem of the section about an alternative characterization of the strong core. 

\begin{theorem}\label{thm:scmatching}
        A matching $M$ of $G$ is in the strong core if and only if it satisfies the following conditions: 
    \begin{itemize}
        \item $|V(M)\cap A|=2$ for every $A \notin \nopath^*$,
        \item $|V(M)\cap V_Q| \geq |V_Q|-1$ for every $Q \in \mathcal{Q}$, except when $Q=\{A\}$ and $A \in \nocyc_0$.
    \end{itemize}
\end{theorem}

\paragraph{An example.} We present two examples in Figure \ref{fig:example1} in order to illustrate these new notions. The solid edges are in $M_0$, while the dashed edges are in $E$. In both examples, $\nocyc=M_0$ because there is no $M_0$-alternating cycle, and $\nocyc_0=\emptyset$. Also, $B\notin \nopath$ in both examples, since there is an $A-B-X$ alternating path, but there is neither an $(A,B;X)$-$\delta$-path nor an $(X,B;A)$-$\delta$-path. It can be checked that the remaining players are in $\nopath$ in both examples.

\begin{figure}[h]
    \centering
    \includegraphics[width=\linewidth]{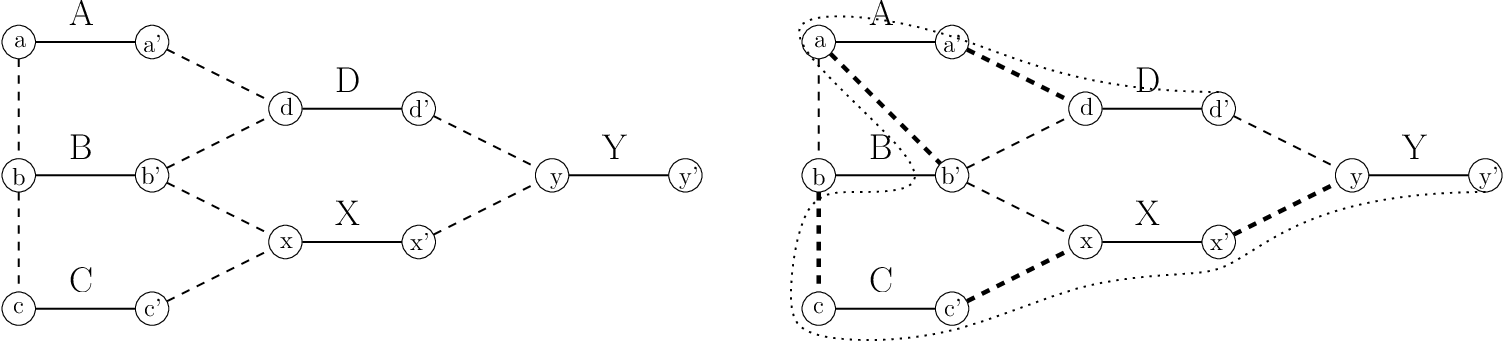}
    \caption{Solid edges are in $M_0$, dashed edges are in $E$. The example on the right is obtained from the example on the left by adding $ab'$. The strong core is empty on the left, while bold edges on the right form a matching in the strong core.}
    \label{fig:example1}
\end{figure}

  In the left example, the pairs are $(A,D),(A,X),(C,D),(C,X),(D,X)$, which means that $\nopath^*=\{A,C,Y\}$, and each are singletons in $G^*$. It can be checked that there is no matching that covers each of $A,C,X$ at least once and the remaining players twice, so $\SC=\emptyset$.

In the right example, $(A,C)$ is also a pair, so $\nopath^*=\{A,C,D,X,Y\}$, where $Y$ is a singleton in $G^*$, while $A,C,D,X$ form a clique. By Theorem \ref{thm:scmatching}, $M \in \SC$ if and only if it covers both vertices of $B$ and it leaves at most one vertex of $Y$ and at most one vertex of $A \cup C \cup D \cup X$ exposed. One such matching $M$ is indicated in bold; the dotted path is the single alternating path in $M \Delta M_0$.

\bigskip

We will see at the end of the section that Theorem~\ref{thm:scmatching} not only allows us to decide in polynomial time whether the strong core is empty, it also enables us to minimize a linear cost function on the strong core in polynomial time. Next, we proceed to prove the theorem.

\subsubsection*{Proving Theorem~\ref{thm:scmatching}}

We start with the key lemma, which states that in many cases, we can either ``unwind'' alternating paths or create a $\delta$-path. In order to be able to refer to the objects in the lemma in various contexts later in the proof, we use subscript or superscript $L$ to denote the objects in the conditions of the lemma. See Figure \ref{fig:unwind} for an illustration.

\begin{figure}[h]
    \centering
    \includegraphics[width=0.9\linewidth]{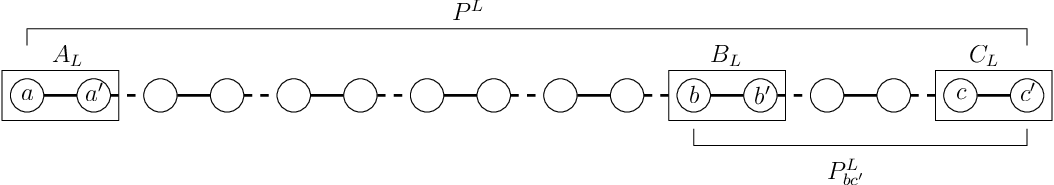}\vskip 1em
    \includegraphics[width=0.9\linewidth]{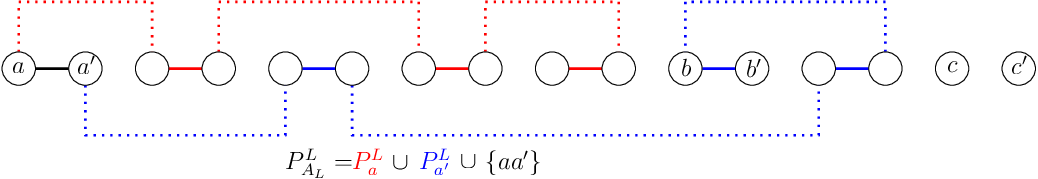}\vskip 1em
    \includegraphics[width=0.9\linewidth]{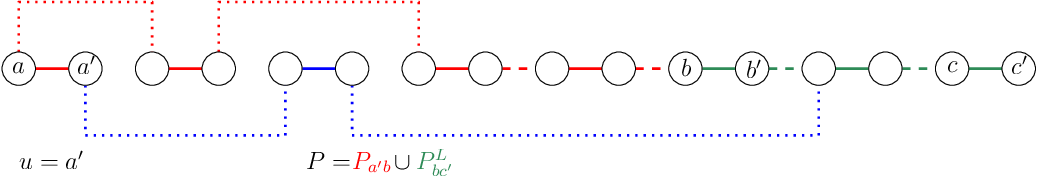}
    \caption{Illustration of the Unwinding Lemma. Top: the $M_0$-alternating $A_L-B_L-C_L$ path $P^L$, and its subpath $P^L_{bc'}$. Middle: The $M_0$-alternating path $P^L_{A_L}$ through the player edge $\pedge_{A_L}=aa'$, and its subpaths $P^L_a$ and $P^L_{a'}$. Bottom: in this example, the lemma gives an $a'-b$ path $P_{a'b}$ such that $P=P_{a'b} \cup P^L_{bc'}$ is an $M_0$-alternating path, and $P^L_{a'}\cup P$ contains an $(A_L,B_L;C_L)$-$\delta$-path $H$ with $P\subseteq H$.}
    \label{fig:unwind}
\end{figure}

\begin{lemma}[Unwinding Lemma]\label{lem:unwind}
    Let $A_L\in \nocyc$, and let $P^L_{A_L}$ be an $M_0$-alternating path through $\pedge_{A_L}$. Furthermore, let $P^L$ be an $A_L-C_L$ $M_0$-alternating path that ends in $c'\in C_L$ for some player $C_L$, and let $A_L\ne B_L=\{ b,b'\}$ be a player whose edge $bb'$ is on this path, where $b'$ is the vertex closer to $c'$. We denote by $P^L_{bc'}$ the part of $P^L$ from $b$ to $c'$.

     Under these conditions, \new{there exist} $u\in A_L$ and an $M_0$-alternating $u-b$ path $P_{ub} \subseteq P^L_{A_L} \cup P^L$ such that 
 $P=P_{ub}\cup P^L_{bc'}$ is an $M_0$-alternating path, and either    
    \begin{itemize}
    \item  the subpath $P^L_u$ of $P^L_{A_L}$ starting in $u$ is disjoint from $P$ (apart from $u$), i.e., $P^L_u\cup P$ is an $M_0$-alternating path, or 
  \item $P^L_u\cup P$ contains an $(A_L,B_L;C_L)$-$\delta$-path $H$ with $P\subseteq H$.
    \end{itemize}

\end{lemma}
\begin{proof}
    Suppose that the conditions of the lemma hold, but $u\in A_L$ and $P_{ub}\subseteq P^L_{A_L}\cup P^L$ cannot be chosen \new{such that the first case holds, i.e.,  $P^L_u\cup P$ is an $M_0$-alternating path. We show that in this case, there must be choices of $u\in A_L$ and $P_{ub}\subseteq P^L_{A_L}\cup P^L$ such that the second case holds.} 
    
    Choose $u\in A_L$ and $P_{ub}\subseteq P^L_{A_L}\cup P^L$ so that $P=P_{ub}\cup P^L_{bc'}$ is an $M_0$-alternating path and $P^L_{u}$ first intersects $P$ (after $u$) as close to $c'$ on $P$ as possible. Such a choice must exist because of the existence of $P^L$ and because of our assumption that $P$ cannot be chosen to be disjoint from $P_u^L$. 

    Let $v$ be this first intersection of $P_u^L\setminus \{u\}$ with $P$, and let $P^L_{uv}$ be the part of $P^L_u$ from $u$ to $v$. \new{ If $v$ is in the part $P^L_{bc'}$ of $P$ and $v$ is the vertex of its player closer to $b$ on $P$, then $P_{uv}^L\cup P$ is an $(A_L,B_L;C_L)$-$\delta$-path $H$ with $P\subseteq H$, and we are done. We show that no other case is possible.}

    \new{Suppose that $v$ is in the part $P^L_{bc'}$ of $P$, but $v$ is the vertex of a player that is closer to $c'$ on $P$. Then $P_{uv}^L\cup P_{uv}$ gives an $M_0$-alternating cycle through $A_L\in \nocyc$, where $P_{uv}$ is the part of $v$ between $u$ and $v$, contradiction.}

    \new{Suppose finally that $v$ is in the part of $P$ between $u$ and $b$.}
    Let $P_{vb}$ be the part of $P$ from $v$ to $b$. 
    Since $A_L\in \nocyc$, we have that $u':=\{ a,a'\} \setminus \{ u\}$ and $P_{u'b}:=\{ u'u\} \cup P^L_{uv}\cup P_{vb}\subseteq P^L_{A_L}\cup P^L$ are also valid choices in the lemma, since if $P_{u'b}$ was not an $M_0$-alternating path (i.e. if the two edges adjacent to $v$ in $P_{u'b}$ were both from $E$), then we would get an $M_0$-alternating cycle \new{$P^L_{uv}\cup P_{uv}$} through $A_L$, a contradiction. Furthermore, as $P^L_{uv}\subseteq P^L_{A_L}$ and $P^L_{u'}\subseteq P^L_{A_L}$ are disjoint, we must have that either $P^L_{u'}$ does not intersect $P'=P_{u'b}\cup P^L_{bc'}$, which contradicts our assumption that it is impossible to choose such a path $P_{u'b}\subseteq P^L_{A_L}\cup P^L$, or $P^L_{u'}$ first intersects $P'$ in the part $P_{vb}\cup P^L_{bc'}\setminus \{ v\}$. Since this part is the same in $P$ and $P'$, we get that this choice of $u'$ and $P_{u'b}$ leads to a case where the first intersection of $P^L_{u'}$ with $P'$ is closer to $c'$ on $P'$ than $v$ was to $c'$ on $P$, in contradiction to the choices of $u$ and $P_{ub}$.
\end{proof}

We remark that if $B_L=C_L$ in Lemma \ref{lem:unwind}, then the path $P$ given by the lemma always satisfies the first property, because there is no $(A_L,B_L;B_L)$-$\delta$-path.

\begin{lemma}\label{prop:nonzero}
    Let $A,B \in \nocyc$, and let $P$ be an $M_0$-alternating $A-B$ path. Then \new{for any} $M \in \SC$ we have $|V(M)\cap A|\ge 1$ and $|V(M)\cap B|\ge 1$.
\end{lemma}

\begin{proof}
    Suppose that $M\in \SC$ is a matching \new{where one of them, say $B$ has 0 vertices covered.  By Lemma~\ref{2strong} this means that $|V(M)\cap A|=2$. } Let $x_A$ and $y_A$ denote the two endpoints of the path $P_A$ containing $\pedge_A$ in $M \Delta M_0$. The existence of $P$ implies that there is an $M_0$-alternating path from $B$ to $P_A$, hence also from $B$ to $x_A$ or $y_A$. Such a path gives a blocking coalition, as $B$ improves and no player gets worse.
\end{proof}

\begin{lemma} \label{lem:cprime}
  Consider an $(A,B;C)$-$\delta$-path $H$ with $C\in \nocyc$, and an $M \in \SC$ such that $|V(M)\cap A|=1$ and $|V(M)\cap C|=2$. Then there is an $(A,B;C')$-$\delta$-path $H'$ with the same cycle as $H$, such that $|V(M)\cap C'|\le 1$ and $C'\in \nocyc^*$.
\end{lemma}

\begin{proof}
    Let $P_C=P_c\cup \{ cc'\} \cup P_{c'}$ be the $M\triangle M_0$ alternating path through $\pedge_C$. Let $P_{dc}=\Hpath (H)$ be the path of the $\delta$-path $H$ that ends in $c\in C$ and starts in $d\in D$ (where $d$ is the common vertex of the path and the cycle of $H$).

We use Lemma~\ref{lem:unwind} as follows. We set $A_L:=C$ and $B_L=C_L:=D$. The path $P^L$ is set to be $P_{dc}$ and $P^L_{A_L}$ is set to be the maximal subpath of $P_C$ that contains $\pedge_C$, but does not intersect $\Hcycle (H)$ (or $\Hcycle (H)\setminus \{ d\}$ if $C=D$). As $\pedge_C\notin \Hcycle (H)$, there is such a subpath. For $u\in \{c,c'\}$, let $P^L_u=P_u \cap P^L_{A_L}$.

Lemma~\ref{lem:unwind} gives us $u\in C$ and $P_{ud'}$ such that \new{the path $P_u^L\setminus \{ u\}$ does not intersect $P_{ud'}\cup \{ dd'\}$ nor $\Hcycle (H)$, by the choice of $P_{A_L}^L$ (if $C=D$ then $u\ne d$ and $u'=d$, so $P_u^L$ is disjoint from $\Hcycle (H)$). We claim that $P_u$ does not intersect $P_{ud'}\cup \{ dd'\}\cup  \Hcycle (H)$, which can be seen as follows. Suppose for the contrary that there is a nomempty intersection. As, $P^L_u$ is a subpath of $P_{u}$ (and one that also starts in $u$), which is disjoint from $\Hcycle (H)\setminus \{ d\}$, and $d\in \Hcycle (H)$, we get that if $P_u\setminus \{ u\}$ intersects $P_{ud'}\cup \{ dd'\} \cup \Hcycle (H)$, then by the choice of $P^L_{A_L}$ the first intersection (from $u$) must be on $\Hcycle (H)$. Let this intersection point be $v$ and $P_{uv}$ be the part of $P_u$ from $u$ to $v$. Then, let $Q_{vd}$ be the subpath of $\Hcycle (H)$ between $v$ and $d$ that starts with $vv'\in M_0$, which must exists if $v\ne d$; if $v=d$ we let $Q_{vd}=\emptyset$. Then, we get an $M_0$-alternating cycle $P_{uv}\cup Q_{vd}\cup \{ dd'\} \cup P_{ud'}$ ($P_{ud'}$ and $Q_{vd}\subseteq \Hcycle (H)$ are disjoint as $P_{ud'}\subseteq (P_{dc}\cup P^L_{A_L})\setminus \{ dd'\}$ by Lemma~\ref{lem:unwind}) through $C\in \nocyc$, a contradiction.}

%    Choose $u\in C$ and the $M_0$-alternating path $P_{du}$ in a way such that $P_{du}$ only intersects the cycle of $H$ in $d$ and either $P_u\setminus \{ u\}$ is disjoint from $P_{du}$ or otherwise the first intersection of $P_u$ (after $u$) with $P_{du}$ is as close to $d$ on $P_{du}$ as possible.

  %  First, we show that $P_u\setminus \{ u\}$ cannot intersect the cycle of $H$ before it intersects $P_{du}$. Indeed, if there is such an intersection, then we obtain an $M_0$-alternating cycle through $C$, contradiction to $C\in \nocyc$. Next, suppose that it intersects $P_{du}$ in a vertex $v$, and let the part of $P_u$ from $v$ to $u$ be $P_{vu}$ and the part of $P_{du}$ from $d$ to $v$ be $P_{dv}$. Then, if we take the choice $u' = \{ c,c'\}\setminus \{ u\}$ and $P_{du'}'=P_{dv}\cup P_{vu}\cup \{ uu'\}$, then $P_{u'}\setminus \{ u'\}$ similarly cannot intersect the cycle of $H$ before it intersects $P'_{du'}$ and either it does not intersect $P_{du'}'$ at all, or it intersects $P_{du'}'$ strictly closer to $d$ than $P_u$ did $P_{du}$ (as $P_{vu}\subseteq P_u$ and $P_{u'}$ are disjoint), both of which contradicts the choices of $u$ and $P_{du}$.

   \new{ Hence, we obtain an $(A,B;C')$-$\delta$-path by taking the union $\Hcycle(H) \cup \{ dd'\} \cup P_{ud'} \cup P_u$, where $C'$ is the last player on $P_u$, and it is in $\mathcal{K^*}$ because $|V(M)\cap C'|\le 1$.}
\end{proof}
\begin{cor}\label{cor:cprime}
If there is an $(A,B;C)$-$\delta$-path and an $M \in \SC$ such that $|V(M)\cap A|=1$, then $|V(M)\cap B|=2$. 
\end{cor}
\begin{proof}
Otherwise, the $(A,B;C')$-$\delta$-path given by Lemma~\ref{lem:cprime} would contain a blocking path for $M$ by Lemma \ref{2strong}.
\end{proof}

\begin{lemma} \label{lem:cases}
Let \new{$Q$} be an $M_0$-alternating $A-B-C$ path for some  $A \in \nocyc$, $B \in \nocyc^*$, $C \in \nocyc$. Then there exists an $(A,B;C)$-$\delta$-path or a $(C,B;A)$-$\delta$-path. 
\end{lemma}

\begin{proof}
Let $M\in \SC$ such that $|V(M)\cap B|=1$ \new{and suppose for the contrary that there exists no $(A,B;C)$-$\delta$-path or $(C,B;A)$-$\delta$-path.} %Then, we must have that $|M\cap A|=2$ or $|M\cap C|=2$. 
Let $P_A$ be the $M\triangle M_0$ path that contains $\pedge_A$ and $P_C$ be the $M\triangle M_0$ path that contains $\pedge_C$.
Let the order of vertices on the $A-B-C$ $M_0$-alternating path \new{$Q$} be $a,a',b,b',c,c'$. 

We use Lemma~\ref{lem:unwind} by setting $A_L,B_L,C_L$ in the lemma as $A,B,C$ with $M_0$-alternating paths \new{$P^L:=Q$ and $P^L_{A_L}:=P_A$, which gives us $u\in \{ a,a'\}$ and $P_{ub}\subseteq P_A\cup Q$.} \new{Let $Q^{\mathrm{new}}=P_{ub}\cup Q_{bc'}$ denote the $A-B-C$ $M_0$-alternating path we get by applying the lemma and possibly replacing the part of $Q$ between $a$ and $b$ because of it. As there is no $(A,B;C)$-$\delta$-path, we get that $Q^{\mathrm{new}}$ is disjoint from $P_u\setminus \{u\}\subseteq P_A$. Note that the part of $Q^{\mathrm{new}}$ and $Q$ between $b$ and $c'$ is the same. } 

\new{Then, we use Lemma~\ref{lem:unwind} again, by setting $A_L,B_L,C_L$ as $C,B,A$ with $P^L:=Q^{\mathrm{new}}$ and $P^L_{A_L}:=P_C$.
The conditions of Lemma~\ref{lem:unwind} are satisfied since $C\in \nocyc$. 
As there is no $(C,B;A)$-$\delta$-path, Lemma~\ref{lem:unwind} gives $w\in C$, and $M_0$-alternating paths $P_{b'w}\subseteq P_C\cup Q^{\mathrm{new}}$ such that $P_w\setminus \{ w\}\subseteq P_C$ is disjoint from $P'=Q^{\mathrm{new}}_{ub'}\cup P_{b'w}$.}

\new{We claim that $P_u\setminus \{u\} \subseteq P_A$ does not intersect $(P'\cup P_w)\subseteq Q^{\mathrm{new}}\cup P_C$. By the above, we know that $P_u\setminus \{ u\}$ is disjoint from $Q^{\mathrm{new}}$. Hence, if $P_u\setminus \{ u\}$ intersects $P'$, then it must be the case that $P_A=P_C$. If $u\in  P_w$, then $P_w\setminus \{ w\}$ is not disjoint from $P'$, contradiction. If $P_w\subseteq P_u$,
then $P_u\setminus \{ u\}$ is not disjoint from $Q^{\mathrm{new}}$, also a contradiction.
Hence, the only possibility left is $P_u\cap P_w=\emptyset$. Then, since the parts of $P'$ and $Q^{\mathrm{new}}$ between $u$ and $b'$ are the same, $P_u\setminus \{ u\}$ must first intersect $P'\cup P_w$ in the part $P'_{b'w}$ of $P'$. If the intersection point $v$ is the vertex of its player that is closer to $w$ on $P'$ then $P'_{uv}\cup P^u_{uv}$ ($P_{uv}^u$ is the part of $P_u$ between $u$ and $v$) is an $M_0$-alternating cycle through $\pedge_A$ and otherwise $P^u_{uv}\cup P'$ is an $(A,B;C)$-$\delta$-path, both are contradictions.}

Hence, we obtain that we can create an $M_0$ alternating path $P_u\cup P'\cup P_w$, where the endpoints both have only one vertex covered by $M$. As $B$ has both vertices covered in $P_u\cup P'\cup P_w$, this leads to a weakly blocking coalition to $M$, contradiction.
\end{proof}
\begin{cor}
    For the subsets of players $\nocyc^*$ and $\nopath$, we have that $\nocyc^*\subseteq \nopath$.
\end{cor}

\begin{lemma} \label{prop:samepath}
    Let $A,B,C\in \nocyc$ and consider an $(A,B;C)$-$\delta$-path $H$, and an $M \in \SC$ such that $|V(M)\cap A|=1$ and $|V(M)\cap B|=2$. Then $\pedge_A$ and $\pedge_B$ are on the same alternating path in $M \Delta M_0$.
\end{lemma}

\begin{proof}
  Suppose that $\pedge_A$ is on the $M\Delta M_0$-alternating path $P_A$ and $\pedge_B$ is on $M\Delta M_0$-alternating path $P_B$ and $P_A\ne P_B$. By $|V(M)\cap A|=1$, $a\in A$ is on one end of $P_A$; let the other end be $d\in D$. Let the endpoints of $P_B$ be $x_B\in X$ and $y_B\in Y$. Using Lemma~\ref{lem:cprime}, we get that we can assume that $|V(M)\cap C|=1$. 

 Let $P_{BC}$ be the $M_0$-alternating path of $H$ starting with $\pedge_B$ and ending with $\pedge_C$ that contains $\pedge_A$. We use Lemma~\ref{lem:unwind} as follows. We set $A_L:= A,B_L=C_L:=D$, $P^L$ as $P_A$, and $P^L_{A_L}$ as taking a maximal subpath of $P_{BC}$ that contains $\pedge_A$, but does not intersect $P_B$ (as $\pedge_A\in P_{BC}$, but $\pedge_A\notin P_B$, there must be such a subpath). Then, Lemma~\ref{lem:unwind} gives a path $P_A'$ disjoint from $P_B$ (neither $P^L$ nor $P^L_{A_L}$ intersects $P_B$) that ends in $d\in D$, contains $\pedge_A$ and either the other end is connected by an edge to $P_B$, or if not, then it ends in $C$ (as $\pedge_B\in P_B$). In both cases, we obtain an $M_0$-alternating path that contains $\pedge_A$ in its interior, and connects two players who both have one vertex covered by $M$.  Thus, the players on this path form a weakly blocking coalition, contradicting $M\in \SC$.
\end{proof}

\begin{lemma} \label{lem:pairs}
    Let $A \in \nocyc^*$, $B,C \in \nopath$ be such that $(A,B)$ and $(A,C)$ are pairs. Then $(B,C)$ is also a pair. 
\end{lemma}

\begin{proof}
    Suppose for the contrary that $(B,C)$ is not a pair. Let $M \in \SC$ be a matching such that $|V(M) \cap A|=1$, and let $P$ be the path in $M \Delta M_0$ that contains $\pedge_A$. By Lemma \ref{prop:samepath}, $\pedge_B$ and $\pedge_C$ are on $P$ and $|V(M)\cap B|=|V(M)\cap C|=2$. We may assume that $P$ is an $A-B-C-D$ path, where $D \in \nocyc^*$ and $|V(M)\cap D|=1$. Let the order of the vertices of these players along the path be $a,a',b,b',c,c',d,d'$.
    
As $(B,C)$ is not a pair, but $C\in \nopath$, we get that there must be a $(C,D;X)$-$\delta$-path $H$, by using the definition of $\nopath$ for the $B-C-D$ subpath of $P$. We may assume that $|V(M)\cap X|\le 1$ and $X\in \nocyc^*$ by Lemma~\ref{lem:cprime}.

We use Lemma~\ref{lem:unwind} as follows. We set $A_L:=D$, $B_L:=C$ and $C_L:=A$. We let $P^L:=P$ and $P^L_{A_L}$ be the $C-D-X$ $M_0$-alternating path of the $\delta$-path $H$. If we get an $M_0$-alternating path through $\pedge_D$ that ends in $A$ and starts in $X$, then that corresponds to a blocking coalition to $M$, contradiction. Hence, we get a $(D,C;A)$-$\delta$-path $H'$, such that it contains the $A-B-C$ subpath of $P$.

As $(A,C)$ is a pair, we know that there is an $(A,C;Y)$-$\delta$-path $H''$ for some $Y\in \nocyc^*$ with $|V(M)\cap Y|\le 1$ by Lemma~\ref{lem:cprime}. 

Finally, we use Lemma~\ref{lem:unwind} for $A_L:=A$, $B_L:=Z$ and $C_L:=D$, where $Z$ is the unique player with the connector vertex in the $(D,C;A)$-$\delta$-path $H'$. We set $P^L$ to be the $A-B-C-D$ subpath of $H'$ and $P^L_{A_L}$ to be the maximal subpath of the $C-A-Y$ $M_0$-alternating path of $H''$, such that it contains $\pedge_A$, but does not intersect the cycle $\Hcycle (H')$ of $H'$. As $\pedge_A\notin \Hcycle (H')$, there is such a subpath. Then, Lemma~\ref{lem:unwind} gives one of two options. In the first one, we get an $Y-A-D$ $M_0$-alternating path, where all three of these players have at most 1 vertex covered by $M$, contradicting $M\in \SC$. In the second option, we get $u\in A$, and path $P_{uz}$, such that the end of $P_u^L$ is connected with an edge to $\Hcycle (H')$ (since the first option does not hold, and $\pedge_C\in P^L$). However, this implies that we obtain an $M_0$-alternating cycle through $\pedge_A$, contradiction.
\end{proof}
\begin{cor}\label{cor:pairs}
      We have that $\nocyc^* \subseteq \nopath^*$.
\end{cor}

\paragraph{Proof of Theorem~\ref{thm:scmatching}.}
    To see that the conditions are necessary, observe that $|V(M)\cap A|=2$ must hold for every $A \notin \nocyc^*$, hence also for $A\notin \nopath^*$ as $\nocyc^*\subseteq \nopath^*$. Also, if $(A,B)$ is a pair, then $|V(M) \cap (A\cup B)|\geq 3$ must hold by Corollary \ref{cor:cprime}. Lastly, if $A \notin \nocyc_0$, \new{then there is a player $B\in \nocyc$ such that there is an $A-B$ $M_0$-alternating path}, hence $|V(M) \cap A| \geq 1$ must hold by Lemma \ref{prop:nonzero}.

To see sufficiency, we use the alternating path characterization in Lemma \ref{2strong} for membership in $\SC$. Suppose $M$ satisfies the two conditions of Theorem~\ref{thm:scmatching}. Then, the properties of $M$ imply that if a player-edge is in an $M_0$-alternating cycle, then both of its vertices are covered by $M$. Let $P$ be an $M_0$-alternating path from $A \in \nocyc$ to $B\in \nocyc$, where none of the vertices of $A$ are covered by $M$. Then $A \in \nocyc_0$, which contradicts that there is an $M_0$-alternating path to some $B\in \nocyc$ from $A$. Finally, suppose that $P$ is an $A-B-C$ alternating path, where $|V(M) \cap B|=1$. If $A\notin \nopath^*$, $B\notin \nopath^*$ or $C\notin \nopath^*$, then one of $A,B,C$ is covered twice by $M$, and we are done. Otherwise, since $A,C\in \nopath^*\subseteq \nocyc$ and $B\in \nopath^* \subseteq \nopath$, the definition of $\nopath$ implies that $(A,B)$ is a pair or $(B,C)$ is a pair. In the former case, $|V(M) \cap A|=2$, while the latter case implies $|V(M) \cap C|=2$ by the second condition on $M$. Thus, the conditions of Lemma \ref{2strong} are satisfied and therefore $M$ is in the strong core.
% Since it can be decided in polynomial time whether an $(A,B;C)$-$\delta$-path exists for given $A,B,C \in \nocyc$ (see Lemma \ref{lem:deltaalg}), membership in $\nopath$ and $\nopath^*$ can be decided in polynomial time. 
%By Lemma~\ref{lem:lowerbound}, existence of a polynomial-time algorithm for finding a member of $\SC$ or determining that $\SC = \emptyset$ follows.
\qed

\medskip 
%Finding a matching with the properties in Theorem \ref{thm:scmatching} can be done in polynomial time by Lemma~\ref{lem:lowerbound}. 
To obtain polynomial-time algorithms, it remains to show that membership in $\nopath$ and $\nopath^*$ can be decided in polynomial time.
\new{Let $\alpha$ denote the inverse Ackermann function. Then, it is known that detecting a negative weight cycle in an undirected graph with integer edge costs bounded by a constant can be done in $\mathcal{O}(|V|^{2.5}\log (|V|)^{1.5}\alpha (|V|^2,|V|)^{0.5})$ time~\cite{williamson2016fast} (Corollary 2). 
Hence, for any constant $\varepsilon>0$ it is true that it can be done in $\mathcal{O}(|V|^{2.5+\varepsilon})$ time. We will use this fact in our next lemma.}

\begin{lemma} \label{lem:deltaalg}
   For given $A,B,C \in \nocyc$, it can be decided in $O(|V|^{2.5+\varepsilon})$ time if an $(A,B;C)$-$\delta$-path exists.  
\end{lemma}
\begin{proof}
As mentioned in the proof of Theorem~\ref{2score-verif}, we can decide in $\mathcal{O}(|V|^{2.5})$ time if there is an $A-B-C$ $M_0$-alternating path for $A,B,C\in \nocyc$. To make sure that $B$ is the middle player, we can iterate through all pairs $(u,v)\in \{ a,a'\} \times \{ c,c'\}$ and delete the edges in $E$ that are adjacent to $u$ and $v$, when searching for an almost-perfect matching.

If there is no such $M_0$-alternating path, then we can conclude that no $(A,B;C)$-$\delta$-path exist. Otherwise, let $P$ be the $A-B-C$ $M_0$-alternating path found, and suppose the order of the vertices is $a,a',b,b',c,c'$.

Suppose that there exists an $(A,B;C)$-$\delta$-path $H$. Using Lemma~\ref{lem:unwind} for $A_L:=A,B_L:=B$ and $C_L:=C$, $P^L:=P$ and the $B-A-C$ path subpath of $H$ as $P^L_{A_L}$, we obtain that there must exist an $(A,B;C)$-$\delta$-path $H'$ that contains the part of $P$ from $\pedge_B$ to $\pedge_C$. Thus, it is enough to decide if there exists such an $(A,B;C)$-$\delta$-path or not.
To do this, first, we create a new graph from $G$ as follows.
\begin{enumerate}
    \item Add $M_0$.
    \item Delete $c'$.
    \item Contract all vertices strictly after $b'$ on $P$ into $c$. Let this new vertex be $\hat{c}$.
    \item Delete all incident edges to $b'$, except $\pedge_B$ and $b'\hat{c}$.
\end{enumerate}

Let the graph obtained be $\hat{G}$. We assign a weight of $-1$ to the edges of $M_0$ in $\hat{G}$, except that we assign $-2$ to $\pedge_A$ and $\pedge_B$. All other edges have a weight of $1$.

We claim that there is a negative weight cycle in $\hat{G}$ if and only if there is an $(A,B;C)$-$\delta$-path in $E\cup M_0$. Since the former can be decided in polynomial time, the statement will also follow from this.

Suppose that there is an $(A,B;C)$-$\delta$-path $H$; we can assume that the part $P_{bc'}$ of $P$ from $b$ to $c'$ is included in $H$. It is easy to see that $(H\setminus P_{bc'})\cup \{\pedge_B,b'\hat{c}\}$ will give a cycle in $\hat{G}$ with weight $-1$, as it alternates between $E$ and $M_0$ except from $\hat{c}$ and contains both $\pedge_A$ and $\pedge_B$.

Conversely, suppose that there is a negative weight cycle in $\hat{G}$. As there is no player edge adjacent to $\hat{c}$, no $M_0$-alternating cycle goes through $\hat{c}$. Hence, as $A,B\in \nocyc$, there cannot be $M_0$-alternating cycles in $\hat{G}$ that contain either $\pedge_A$ or $\pedge_B$. Since $M_0$ forms a matching, any negative weight cycle must contain $\pedge_A$ or $\pedge_B$. Furthermore, any negative weight cycle that is not $M_0$-alternating must contain both. Hence, by the previous observation, there must be a cycle in $\hat{G}$ that is $M_0$-alternating except at one vertex $v$, and contains $\pedge_A$ and $\pedge_B$. Since $b'$ has only $\pedge_B$ and $b'\hat{c}$ as incident edges, we must have that $v=\hat{c}$. Then, by $A,B\in \nocyc$, it is easy to see that this cycle gives an $(A,B;C)$-$\delta$-path in $E\cup M_0$ (by adding back $\pedge_C$ and undoing the contractions). 

As finding a negative cycle in an undirected graph with edge weights $\{ -2,-1,1\}$ can be done in $O(|V|^{2.5+\varepsilon})$ time~\cite{williamson2016fast}, the lemma follows.
\end{proof}

 The bottleneck for the running time of the algorithms is the verification of membership in $\nopath$ and $\nopath^*$, and the construction of the graph $G^*$. We did not try to optimize the running time of these subroutines; the naive approach requires $O(|V|^3)$ maximum size matching algorithms and negative cycle detections in an undirected graph with edge costs in $\{-2,-1,1\}$. Since both of these can be done in time $O(|V|^{2.5+\varepsilon})$~\cite{blum1990new,williamson2016fast}, we get a total running time of $O(|V|^{5.5+\varepsilon})$. 
\begin{cor}\label{cor:scoredecision}
\strongex-2 is solvable in $O(|V|^{5.5+\varepsilon})$ time.
\end{cor}

\paragraph{Optimization on the strong core.}
An interesting and useful consequence of Theorem \ref{thm:scmatching} is that for given edge costs $c \in \mathbb{Q}^E$, we can find a minimum cost matching in the strong core. This is in stark contrast to the case of the weak core, where optimization is NP-hard even for players of size 1, as shown in Theorem \ref{thm:1opthard}. Indeed, by Theorem~\ref{thm:scmatching}, the strong core can be characterized by setting lower bounds on $|U_i\cap V(M)|$ for some partition $U_1\cup \cdots \cup U_l$, so by Lemma~\ref{lem:lowerbound}, after the partition and the bounds are computed, a minimum cost such matching can be found in $O(|V|^{4})$  time, if one exists.
\begin{cor}\label{cor:scoreopt}
If the strong core is nonempty, then for any edge costs $c \in \mathbb{Q}^E$, we can find a minimum-cost matching in the strong core in $O(|V|^{5.5+\varepsilon})$ time.   
\end{cor}

\section{Small number of players}\label{sec:constno}

In this section, we show that if the number of players is small, i.e., bounded by a constant, then the existence of a strong (resp. weak) core solution, the verification of membership in the cores, and optimization on the cores are all solvable in polynomial time, with an algorithm that is only exponential in the number of players. 
This is relevant for example in the IKEP application, where players correspond to countries, which \new{often} implies a fairly small upper bound on the number of players, \new{usually between 3 and 8}~\cite{biro2019building,valentin2019international}.  We shall see that the tractability of this case is a direct consequence of Lemma \ref{lem:lowerbound}.

Recall that the instance is given by a graph $G=(V,E)$ and a partition $\partition=(V_1, V_2 , \dots , V_m)$ of $V$, where $m$ is the number of players. In this section, $m$ is regarded as a small constant, but our results are true for arbitrary $m$. However, the running times become infeasible for large $m$, because of the exponential dependence on $m$.

First, we consider the verification problems.  Lemma \ref{lem:lowerbound} immediately implies that membership of a given matching $M$ in the weak core can be decided in $\mathcal{O}(2^m|V|^{2.5})$ time. Indeed, the number of possible coalitions is $2^m$, and for a given coalition
$\{ i_1, i_2,\dots,i_k\} $, we can decide, using Lemma \ref{lem:lowerbound}, if there is a matching $M'$ in the
induced subgraph $G[V_{i_1} \cup V_{i_2}\cup \dots \cup V_{i_k}]$ such that $|V(M') \cap V_{i_j}| \geq |V(M) \cap V_{i_j}|+1$ for every $j \in [k]$ in $\mathcal{O}(|V|^{2.5})$ time.

Membership in the strong core can be decided similarly, but for every coalition
$\{ i_1, i_2,\dots,i_k\}$ and every $\ell\in [k]$, we check if there is a matching $M'$ in the
induced subgraph $G[V_{i_1} \cup V_{i_2}\cup \dots \cup V_{i_k}]$ such that $|V(M') \cap V_{i_j}| \geq |V(M) \cap V_{i_j}|$ for every $j\in [k]$ and $|V(M') \cap V_{i_\ell}| \geq |V(M) \cap V_{i_\ell}|+1$.
 This leads to a runtime of $\mathcal{O}(2^m|V|^{3.5})$.

\begin{cor}\label{XP-verif}
\weakverif\ can be solved in $\mathcal{O}(2^m|V|^{2.5})$ time, and \strongverif\ can be solved in $\mathcal{O}(2^m|V|^{3.5})$ time.
\end{cor}

Deciding non-emptiness of the weak and strong cores is somewhat more difficult, but still polynomial-time solvable if $m$ is constant. The crucial observation is that the membership of a matching $M$ in these cores depends only on the values $|V(M)\cap V_i|$ ($i \in [m]$). Let us call an integer vector \emph{maximal} with respect to a property if the property is violated by increasing any component of the vector by 1.  Using Lemma \ref{lem:lowerbound}, we can find in $\mathcal{O}(|V|^{m+2.5})$ time all maximal vectors $x \in \Zset_+^m$ for which a matching $M$ with $|V(M) \cap V_i|\ge x_i$ for every $i \in [m]$ exists. Indeed,  $x_i \leq |V_i|$ ($i \in [m]$) can be assumed, so we can check all possible $\mathcal{O}(|V|^m)$ vectors, each in $\mathcal{O}(|V|^{2.5})$ time. Let $\mathcal{X}$ denote the set of such maximal vectors.
Since for any matching $M$ covering a set $U\subseteq V$ of vertices, there exists a maximum size matching $M'$ that also covers $U$, and so $|V(M')\cap V_i|\ge |V(M)\cap V_i|$ $\forall i\in [m]$, we can assume that for each $x\in \mathcal{X}$, the sum of its coordinates is two times the size of a maximum matching. Hence, $|\mathcal{X}|\le \mathcal{O}(|V|^{m-1})$

It is easy to see that if the weak (resp. strong) core is nomempty, then it contains a matching $M$, for which $x_i=|V(M)\cap V_i|$ gives a maximal vector $x\in \mathcal{X}$ (else, there is a matching, where everyone is weakly better off, and some players are strictly better off -- so a strongly (resp. weakly) blocking coalition to the latter would also block $M$). 
By Corollary~\ref{XP-verif}, for each $x\in \mathcal{X}$, we can check in $\mathcal{O}(2^m|V|^{3.5})$ time, whether \emph{every} matching $M$ with $|V(M) \cap V_i| =  x_i$ ($i \in [m]$) is in the strong or weak core. If the answer is negative for every $x \in \mathcal{X}$, then the core is empty; otherwise, it is not empty.

\begin{cor}\label{XP-ex}
   \weakex\ and \strongex\ are solvable in $\mathcal{O}(|V|^{m+2.5})$ time. Furthermore, a matching in the weak or strong core can be found in the same time, if one exists.
\end{cor}

Optimization on the weak and strong cores can be done by similar ideas.

\begin{theorem}
Given an NTU partitioned matching game with $m$ players, and a cost vector $c \in \mathbb{Q}^E$, we can find a minimum cost matching in the weak (resp.\ strong) core in $\mathcal{O}(|V|^{m+4})$ time. 
\end{theorem}

\begin{proof}
    For a given vector $x \in \Zset_+^m$ for which $x_i \leq |V_i|$ ($i \in [m]$), we can decide using Lemma \ref{lem:lowerbound} in $\mathcal{O}(2^m|V|^{2.5})$ time whether there is a coalition $\{i_1,\dots,i_k\}$ and a matching $M$ in the
induced subgraph $G[V_{i_1} \cup V_{i_2}\cup \dots \cup V_{i_k}]$ such that $|V(M) \cap V_{i_j}| \geq x_{i_j}+1$ for every $j \in [k]$. Let $\mathcal{X}_w$ denote the set of vectors for which there is no such coalition. Then, $|\mathcal{X}_w|\le \mathcal{O}(|V|^{m})$. Similarly, we can decide in $\mathcal{O}(2^m|V|^{3.5})$ time whether there is a coalition $\{i_1,\dots,i_k\}$, an index $\ell \in [k]$, and a matching $M$ in the
induced subgraph $G[V_{i_1} \cup V_{i_2}\cup \dots \cup V_{i_k}]$ such that $|V(M) \cap V_{i_j}| \geq x_{i_j}$ for every $j \in [k]$ and $|V(M) \cap V_{i_{\ell}}| \geq x_{i_{\ell}}+1$. Let $\mathcal{X}_s$ denote the set of vectors for which there is no such coalition. Then, $|\mathcal{X}_s|\le \mathcal{O}(|V|^{m})$. 

For $x\in \mathcal{X}_w$, let $\mathcal{M}_x$ denote the set of matchings in $G$ satisfying $|V(M) \cap V_i|\ge x_i$ for every $i \in [m]$. Note that $x\in \mathcal{X}_w$ means that every matching in $\mathcal{M}_x$ is in the weak core; furthermore, every matching in the weak core is in $\mathcal{M}_x$ for some $x\in \mathcal{X}_w$. 
Then, we can find a minimum cost matching in $\mathcal{M}_x$ for each $x\in \mathcal{X}_w$ in $\mathcal{O}(|V|^4)$ time by Lemma~\ref{lem:lowerbound}. This, in total, gives a runtime of $\mathcal{O}(|V|^{m+4})$.

The strong core case is analogous.

%We can find a minimum cost matching in $\mathcal{M}_x$ the following way. For $i \in [m]$, let $t_i=|V_i|-x_i$. We construct a graph $G'$ and edge costs $c'$ from $G$ and $c$ by adding vertices $v^1_i, v^2_i,\dots, v^{t_i}_i$ and connecting these vertices to all vertices in $V_i$, for each $i \in [m]$. We also add an additional vertex if the number of vertices is odd, and add a clique on all the new vertices. The new edges have cost 0, while the original edges retain their cost. 

%It is easy to see that if $M'$ is a perfect matching in $G'$, then $M'\cap E$ is a matching in $\mathcal{M}_x$ of the same cost, and conversely, any matching $M\in \mathcal{M}_x$ can be extended to a perfect matching of the same cost in $G'$. Thus, we can efficiently find a minimum cost matching in $\mathcal{M}_x$ for each $x\in \mathcal{X}_w$, and consequently we can find the minimum cost matching in the weak core in polynomial time.

%Finding the minimum cost matching in the strong core can be done using exactly the same method as for the weak core; the only difference is that we replace $\mathcal{X}_w$ by $\mathcal{X}_s$ in the algorithm above.
\end{proof}

\section{Hardness results for partition classes of constant size}\label{sec:hard}

In contrast to the case where the number of players is bounded by a constant, bounding the sizes of the partition classes by a constant larger than 2 does not lead to polynomial-time solvable problems. In this section, we show hardness results of this type. We have seen in Section \ref{sec:size2} that our problems are tractable when all partition classes have size at most 2. Here we show that having partition classes of size 3 leads to coNP-hardness for most problems. 

\subsection{Weak core}\label{weak}

Our first result is the hardness of verifying if a matching belongs to the weak core. 

\begin{theorem}\label{thm:weakdecision}
    \weakverif-3 is coNP-complete, even if the graph is bipartite.
\end{theorem}
\begin{proof}
    We reduce from the NP-complete problem \textsc{x3c} ({\sc exact cover by 3-sets}). An instance of \textsc{x3c} consists of $3n$ elements $\mathcal{X}=\{ 1,2,\dots,3n\}$ and a family of 3-sets $\mathcal{S}=\{ S_1,S_2,\dots,S_{m}\}$ and the question is whether there exists $\mathcal{S'}\subseteq \mathcal{S}$ such that the sets in $\mathcal{S'}$ cover each element exactly once. 

    Let $I$ be an instance of \textsc{x3c}. We create a bipartite instance $I'$ of the \weakverif-3 as follows. 
First, we describe the vertex set $V$ of the graph.
    \begin{itemize}
        \item [--] For each element $i\in [3n]$ we create a vertex $a_i$.
        \item [--] We also create connector vertices $b_{1,2}^1,b_{1,2}^2,b_{2,3}^2,b_{2,3}^3\dots,b_{3n-1,3n}^{3n-1},b_{3n-1,3n}^{3n} $.
        \item [--] For each set $S_j\in \mathcal{S}$, we create vertices $c_j^1,c_j^2,c_j^3$ and vertices $d_j^1,d_j^2$. 
        
    \end{itemize}

    We now describe the edges of the graph.
    \begin{itemize}
        \item [--] For each $j\in [m]$, $l\in [2]$ we have an edge $c_j^ld_j^l$.
        \item [--] For each $i\in [3n-1]$, we have edges $b_{i,i+1}^ib_{i,i+1}^{i+1}$.
        \item [--] For each set $S_j=\{ j_1,j_2,j_3\}$ with $j_1<j_2<j_3$ we have edges $c_j^1a_{j_1},c_j^2a_{j_2},c_j^3a_{j_3}$.
    \end{itemize}

    The graph is bipartite, with bipartition $U=\{ c_1^1,c_1^2,c_1^3,\dots,c_m^1,c_m^2,c_m^2\} \cup \{ b_{i,i+1}^i\mid i\in [3n-1]\}$ and $W=V\setminus U$.

    Finally, we describe the partition of the players. For each $j\in [m]$, we have a player $C_j=\{ c_j^1,c_j^2,c_j^3\}$ and a player $D_j=\{ d_j^1,d_j^2\}$. For each $i\in \{ 2,\dots,3n-1\}$ we have a player $A_i=\{ b_{i-1,i}^i,a_i,b_{i,i+1}^i\}$. In addition, we have players $A_1=\{ a_1,b_{1,2}^1\}$, and $A_{3n}=\{ b_{3n-1,3n}^{3n},a_{3n}\}$. Clearly, all players have size at most 3.

    We create a matching $M$ as follows. We add $c_j^ld_j^l$ for $j\in [m],l\in [2]$ and $b_{i,i+1}^ib_{i,i+1}^{i+1}$ for $i\in [3n-1]$. The construction is illustrated in Figure~\ref{fig:weakverif}

To show coNP-hardness, we show that $I$ admits an exact 3-cover if and only if $M$ is not in the weak core.

    \begin{claim}
        If there is an exact 3-cover in $I$, then $M$ is not in the weak core.
    \end{claim}
    \begin{proof}
        Let $\{ S_{l_1},\dots S_{l_n}\}$ be an exact 3-cover in $I$. We show that the players $\{ C_{l_1},\dots,C_{l_n},A_1,\dots,A_{3n}\} $ form a strongly blocking coalition. As $\{ S_{l_1},\dots S_{l_n}\}$ is an exact cover, there must be a perfect matching between the included $c_j^l$ and $a_i$ vertices. Also, we can include all edges of the form $(b_{i,i+1}^i,b_{i,i+1}^{i+1})$ for $i=1,\dots,3n-1$. Hence, there is a perfect matching on the vertex set of these players. As each of these players has a vertex exposed in $M$ ($c_j^3$ for $C_j$, $a_i$ for $A_i$), it follows that $M$ is not in the weak core.
    \end{proof}
    \begin{claim}
        If $M$ is not in the weak core, then there is an exact 3-cover in $I$.
    \end{claim}
    \begin{proof}
        Suppose that $M$ is not in the weak core. This means that there is a strongly blocking coalition $\mathcal{P}$ to $M$ with a matching $M'$. We can assume that there is a player $A_i\in \mathcal{P}$. Otherwise, there could be only $C_j$ players in $\mathcal{P}$ (the $D_j$ players have all vertices covered by $M$), and they induce no edges between them, a contradiction.

        Observe that players $A_1,\dots,A_{3n}$ must all be in $\mathcal{P}$, because for $A_i$ to strictly improve, $b_{i-1,i}^{i}$ and $b_{i,i+1}^i$ (if they exist) must be matched, so $b_{i-1,i}^{i-1}b_{i-1,i}^{i},b_{i,i+1}^{i}b_{i,i+1}^{i+1}\in M'$. This implies that $A_{i-1},A_{i+1}\in \mathcal{P}$ and similarly we get that $A_1,\dots,A_{3n}$ are all in $\mathcal{P}$.

        Finally, as each player must strictly improve, every $a_i$ vertex is matched in $M'$ and so they are matched to $C_j$ players. Whenever a $C_j$ player participates in the blocking coalition by matching one of its vertices to an $a_i$ vertex, $\{ c_j^1a_{j_1},c_j^2a_{j_2},c_j^3a_{j_3} \} \subseteq M'$, because $D_j\notin \mathcal{P}$ and $C_j$ must strictly improve. By combining this with the fact that each $a_i$ is matched, we obtain that these players must correspond to an exact 3-cover. 
    \end{proof}
The statement of the theorem follows.
\end{proof}

\begin{figure}
    \centering
    \includegraphics[width=0.5\linewidth]{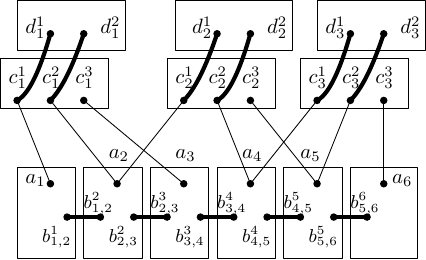}
    \caption{Illustration for Theorem~\ref{thm:weakdecision}. The figure depicts the instance $I'$ corresponding to an {\sc x3c} instance with $S_1=\{ 1,2,3\} , S_2=\{ 2,4,5\}$ and $S_3=\{ 4,5,6\}$. Bold edges denote the edges of $M$.}
    \label{fig:weakverif}
\end{figure}

Intuitively, it would be tempting to conjecture that the weak core is always nonempty in any NTU partitioned matching game.
However, the following example (communicated to us by Zsuzsanna Jankó) shows that this is not the case.

\begin{example}[Empty weak core]\label{ex}
Let $A=\{a_1,a_2,\dots,a_7\}$, $B=\{b_1,\dots,b_7\}$, and $C=\{c_1,\dots,c_7\}$ be three players each having seven vertices. The graph consists of five disjoint complete graphs, three $K_5$'s: $\{a_1,c_2,c_3,b_4,b_5\}$, $\{b_1,a_2,a_3,c_4,c_5\}$ and $\{c_1,b_2,b_3,a_4,a_5\}$, and two $K_3$'s: $\{a_6,b_6,c_6\}$ and $\{a_7,b_7,c_7\}$.

The weak core of this instance is empty. Indeed, if there were a matching in the weak core, then there would also be a maximum-size matching in it. Let $M$ be any maximum-size matching, covering 16 vertices. We may assume that every player has at least 4 vertices covered, since otherwise a single player would block $M$. The two players who have the least number of vertices covered have together at most $\frac{2}{3}\times 16<11$ vertices covered, and therefore, they have at most 10. If both of them have 5 covered in $M$, then they form a blocking coalition because there is a matching where both of them have 6 vertices covered. If one of them has 4 vertices covered while the other 6, then they form a blocking coalition because there is a matching where the first player has 5 vertices covered while the second has 7. Thus, every matching has a strongly blocking coalition.
\end{example}

We use Example~\ref{ex} in our reductions to prove that deciding the emptiness of the weak core is hard.

%\textcolor{red}{Instead of a special edge, I changed to using a \NO-gadget. (\NO\ is a macro)}

\paragraph{The \NO-gadget.} We introduce the key gadget in the proof, the \emph{\NO-gadget}. 
A \NO-gadget between $u$ and $v$ will ensure that (i) neither $u$, nor $v$ can be connected to the \NO-gadget in a matching in the weak core, but (ii) if a matching satisfies (i) for every \NO-gadget and contains appropriate edges from each \NO-gadget, then the matching is in the weak core if and only if it is in the weak core of the instance obtained by substituting each \NO-gadget (between some vertices $u',v'$) by an edge ($u'v'$). Thus, we can use \NO-gadgets to simulate forbidden edges.

Specifically, a \NO-gadget between $u$ and $v$ is constructed as follows. Let $H_{uv}$ be a copy of the instance in Example~\ref{ex} with players $A_{H_{uv}},B_{H_{uv}},C_{H_{uv}}$ and let $a_1^{H_{uv}}$ be a vertex of $A_{H_{uv}}$ in the clique of size 5 which contains only a single vertex of $A_H$. We add a player $S^{uv}$ with four vertices $s_1^{uv},s_2^{uv},s_3^{uv}$ and $s_4^{uv}$, and a player $T^{uv}$ with two vertices $t_1^{uv}$ and $t_2^{uv}$. 

The edges are the following. We keep all edges of $H_{uv}$ and add $us_1^{uv},s_1^{uv}t_1^{uv},s_2^{uv}s_3^{uv},s_2^{uv}a_1^{H_{uv}},s_4^{uv}t_2^{uv},s_4^{uv}v$. See Figure \ref{specialfig} for an illustration.

\begin{figure}[h]
\begin{center}
\includegraphics[scale=.6]{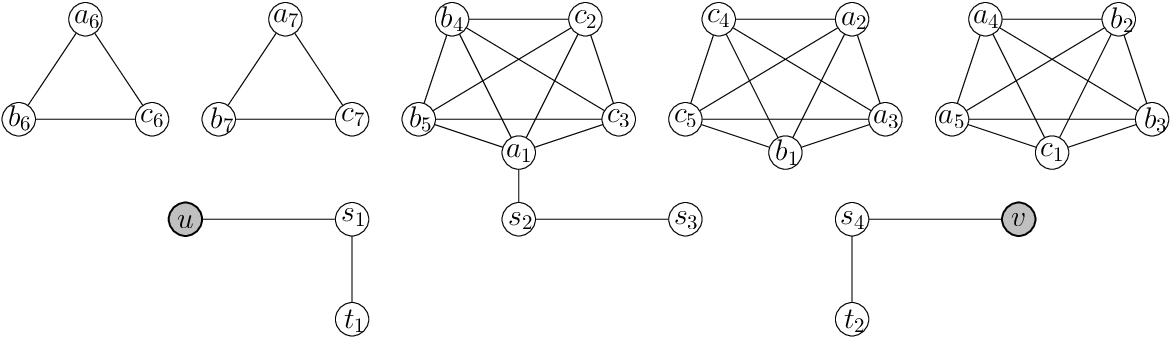}
\caption{The \NO-gadget between $u$ and $v$. Upper indices have been omitted, since the figure illustrates only a single \NO-gadget.}
\label{specialfig}
\end{center}
\end{figure}

%\begin{lemma}\label{special}
%Suppose that an instance $I$ of the NTU partitioned matching game contains a  \NO-gadget between $u$ and $v$. Then, the following hold.
%\begin{enumerate}[i)]
%\item For any matching $M$ in the weak core, the edges $s_2a_1^H$, $s_1t_1$ and $s_4t_2$ belong to $M$. This means that $us_1$ and $s_4v$ cannot belong to $M$.

%\item Suppose the edges $us_1$ and $s_4v$ do not belong to a matching $M$. If in the instance $I'$ obtained by replacing the \NO-gadget between $u$ and $v$ with an edge $uv$, the restriction of $M$ is not in the weak core of $I'$, then $M$ is not in the weak core of $I$.
%\end{enumerate}
%\end{lemma}

%Let \weakex\ denote the problem of deciding if an instance of the NTU Partitioned Matching game admits a matching in the weak core or not. Furthermore, 
Let \weakexforb\ denote the problem, which asks if there is a matching in the weak core that excludes a set $F$ of forbidden edges.

\begin{lemma}\label{special}
   \weakexforb\ with players of size at most $d$ can be reduced to \weakex\ with players of size at most $\max \{ 7,d\}$. 
\end{lemma}
\begin{proof}

Let $I=(G,\partition, F)$ be an instance of \weakexforb\ with players of size at most $d$. We create an instance $I'=(G',\partition')$ of \weakex\ by replacing each edge $uv \in F$ with a \NO-gadget between $u$ and $v$. As players in the \NO-gadget have size at most 7, $I'$ has players of size at most $\max \{ d,7\}$.

We claim that $I$ admits a weak core matching disjoint from $F$ if and only if $I'$ admits a matching in the weak core.

Suppose that $I'$ has a matching $M'$ in the weak core. Let $M=M' \cap E(G)$.
As we have seen in Example~\ref{ex}, for any \NO-gadget, the copy $H_{uv}$ of Example~\ref{ex} admits no weak core matching without additional edges, so $s_2^{uv}a_1^{H_{uv}}\in M'$ for all \NO-gadgets. Then, we also have that $s_1^{uv}t_1^{uv}\in M'$ and $s_4^{uv}t_2^{uv}\in M'$, else $\{ T^{uv},S^{uv}\}$ is a strongly blocking coalition to $M'$ with matching $\{ s_1^{uv}t_1^{uv},s_2^{uv}s_3^{uv},s_4^{uv}t_2^{uv}\} $. By this observation, any edge of $M'$ incident to a vertex $v\in V(G)$ is also present in $M$, so $|V_i\cap V(M)|=|V_i\cap V(M')|$ for all $V_i\in \partition$. Furthermore, it also follows that $M\cap F=\emptyset$. 

Suppose for the contrary that $M$ is not in the weak core of $I$. Let $\mathcal{P}$ be a strongly blocking coalition to $M$ with some matching $M_{\mathcal{P}}$. We claim that $\mathcal{P}\cup \{ S^{uv}\mid uv\in M_{\mathcal{P}}\cap F\}$ strongly blocks $M'$ with the matching $M_{\mathcal{P}}\setminus F\cup \{ us_1^{uv}, s_2^{uv}s_3^{uv},s_4^{uv}v\mid uv \in M_{\mathcal{P}}\cap F\}$. Indeed, as $|V_i\cap V(M)|=|V_i\cap V(M')|$ for all $V_i\in \partition$, all players $i\in \mathcal{P}$ improve from $M'$. Furthermore, the participating $S^{uv}$ players get all their vertices matched, while $s_3^{uv}$ was unmatched in $M'$ (as $s_1^{uv}a_1^{H_{uv}}\in M'$). This contradicts that $M'$ is in the weak core.

Assume that $I$ admits a matching $M$ in the weak core such that $M\cap F=\emptyset$. Then, we can define a matching $M'$ in $I'$ as $M\cup \{s_1^{uv}t_1^{uv},s_2^{uv}a_1^{H_{uv}},s_4^{uv}t_2^{uv},b_6^{H_{uv}}c_6^{H_{uv}},b_7^{H_{uv}}c_7^{H_{uv}},a_i^{H_{uv}}a_{i+1}^{H_{uv}},b_i^{H_{uv}}b_{i+1}^{H_{uv}},c_i^{H_{uv}}c_{i+1}^{H_{uv}}  \mid uv\in F, i=2,4\}$. See Figure~\ref{Mspec} for an illustration.

We claim that $M'$ is in the weak core of $I'$. Suppose that $\mathcal{P}'$ blocks $M'$ with a matching $M'_{\mathcal{P}'}$. First, observe that no player $T^{uv}$ can be in $\mathcal{P}'$, since both vertices are matched in $M'$. Hence, if $S^{uv}\in \mathcal{P}'$ for some $uv\in E(G)$, then all of its vertices must be matched, so for any $uv\in F$ either we have that $\{ us_1^{H_{uv}},s_4^{H_{uv}}v\}\subseteq  M'_{\mathcal{P}'}$ or $\{ us_1^{H_{uv}},s_4^{H_{uv}}v\}\cap  M'_{\mathcal{P}'}=\emptyset$.

Let $M_{\mathcal{P}}$ be the matching we get in $I$ by deleting any edge from a \NO-gadget in $M'_{\mathcal{P}'}$, and adding the edge $uv$ if $\{ us_1^{H_{uv}},s_4^{H_{uv}}v\}\subseteq  M'_{\mathcal{P}'}$. Then, $M_{\mathcal{P}}$ is a matching and for any $V_i\in \partition$, $|V_i\cap M_{\mathcal{P}}|=|V_i\cap M'_{\mathcal{P}'}|$. As $|V_i\cap M|=|V_i\cap M'|$ also holds for $V_i\in \partition$, this gives a contradiction, unless $M_{\mathcal{P}}=\emptyset$ and $\mathcal{P}'$ contain only players in \NO-gadgets. Then, no $S^{uv}$ player can be in $\mathcal{P}'$ either, as $s_1^{uv},s_4^{uv}$ cannot be matched.% as $T^{uv}\notin \mathcal{P'}$.

Hence, it only remains to show that the players of some copy $H_{uv}$ of Example~\ref{ex} cannot form a strongly blocking coalition. It is easy to see that all three of them cannot form a strongly blocking coalition. Also, no player among them can match $6$ of their vertices alone to improve. Finally, for any two players, at least 2 of their vertices will remain unmatched in any matching, but at most $3$ of them are unmatched in $M'$, so they cannot both improve.
\end{proof}
\begin{figure}[h]
\begin{center}
\includegraphics[scale=.6]{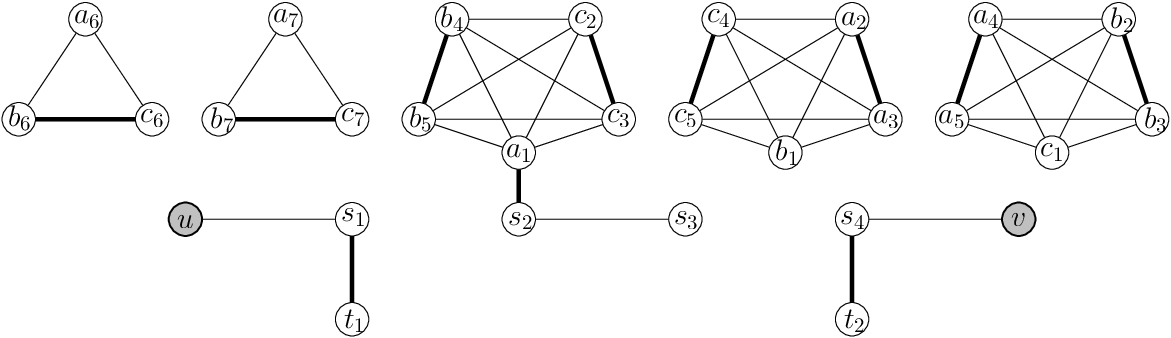}
\caption{The edges in a \NO-gadget in $M'$ in the proof of Lemma~\ref{special}. Bold edges belong to $M'$.}
\label{Mspec}
\end{center}
\end{figure}

\begin{theorem}\label{nphard}
\weakex-7 is NP-hard.
\end{theorem}
\begin{proof}
By Lemma~\ref{special}, it is enough to show NP-hardness of \weakexforb, with players of size at most 7.

We reduce from 3-SAT, where we can assume that each variable appears at most four times. Given an instance $I$ of 3-SAT, we construct an instance $I'$ of the \weakexforb. For every variable of $I$, we construct a variable gadget, and for every clause we construct a clause gadget. After that, we describe the interconnecting edges between the clause gadgets and the variable gadgets.

\textbf{Variable gadget.} For a variable $x_i$, we define four sets of vertices, $X_i$, $\bar{X_i}$, $Y_i$ and $\bar{Y_i}$. $X_i$ contains a vertex for every occurrence of the variable $x_i$ in unnegated form, and the vertices of $X_i$ belong to one player. $\bar{X_i}$ contains a vertex for every occurrence of the variable $x_i$ in negated form, and the vertices of $\bar{X_i}$ belong to one player. For every vertex in $X_i$, there is a separate vertex in $Y_i$ so that there is an edge between them, and each vertex of $Y_i$ belongs to a separate player of size one. Similarly, for every vertex in $\bar{X_i}$, there is a separate vertex in $\bar{Y_i}$ so that there is an edge between them, and each vertex of $\bar{Y_i}$ belongs to a separate player of size one. Every vertex of $Y_i$ is connected to every vertex of $\bar{Y_i}$ by a forbidden edge in $F$. See Figure \ref{variable}.

\begin{figure}[h]
\begin{center}
\includegraphics[scale=.4]{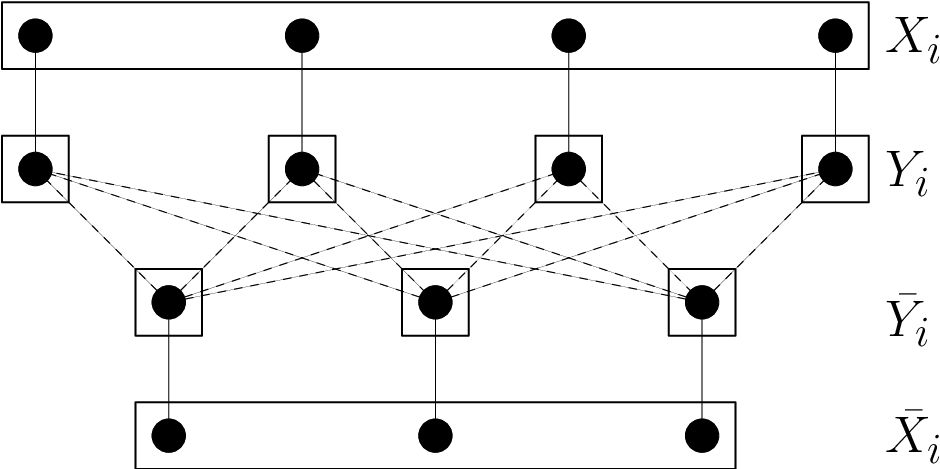}
\caption{The variable gadget for $x_i$. The dotted lines represent forbidden edges.}
\label{variable}
\end{center}
\end{figure}

\textbf{Clause gadget.} For each clause $c_j$, we associate a copy of the instance in Example \ref{ex}. We use the same notation for the vertices as in the example, but every vertex gets an upper index $j$. Let $a_1^j$, $b_1^j$ and $c_1^j$ correspond to the three literals in clause $c_j$.

\textbf{Interconnecting edges.} Now we describe the edges between the clause gadgets and the variable gadgets. To every literal that appears in a clause, there is a corresponding vertex in the clause gadget and a corresponding vertex in the variable gadget. We connect these two vertices with an edge.

As every variable appears at most four times, all players have size at most 7.

\begin{claim}
For any matching $M$ in the weak core of $I'$, there is a truth assignment that satisfies $I$.
\end{claim}

\begin{proof}
If $M$ is in the weak core, then $Y_i$ or $\bar{Y_i}$ has all its vertices covered by $M$. Indeed, if there were a vertex $u\in Y_i$ and a vertex $v\in\bar{Y_i}$ so that neither is covered by $M$, then the players $\{u\}$ and $\{v\}$ together with the forbidden edge connecting them would form a strongly blocking coalition (see Lemma \ref{special}). Since forbidden edges cannot belong to $M$, every vertex of $X_i$ (or $\bar{X_i}$) is matched to a vertex in $Y_i$ (or $\bar{Y_i}$). In the first case, we set the variable $x_i$ to be false, and in the second case we set it to be true. If both hold simultaneously, we arbitrarily set $x_i$ to be true or false, but it still holds that if $x_i$ is set to false, every vertex in $X_i$ is matched to a vertex in $Y_i$ in $M$.

A clause gadget is a copy of the instance in Example \ref{ex}, therefore without the additional (interconnecting) edges, it cannot admit a matching in the weak core, thus $M$ has to contain an interconnecting edge leaving this clause gadget. The interconnecting edge corresponds to a literal that cannot be set to false, because then the vertex corresponding to it in $X_i$ or $\bar{X_i}$ for some $i$ would be matched to a vertex in $Y_i$ or $\bar{Y_i}$ in $M$. This means that in this truth assignment, every clause contains a literal set to true, thus it satisfies $I$.
\end{proof}

\begin{claim}
For any truth assignment satisfying $I$, it is possible to construct a matching $M$ in the weak core of $I'$ without any forbidden edges.
\end{claim}

\begin{proof}
If a variable $x_i$ is set to true (false), let $M$ contain all the edges between $\bar{X_i}$ and $\bar{Y_i}$ (between $X_i$ and $Y_i$), and all interconnecting edges incident to $X_i$ ($\bar{X_i}$), see Figure \ref{Mvari}.

\begin{figure}[h]
\begin{center}
\includegraphics[scale=.4]{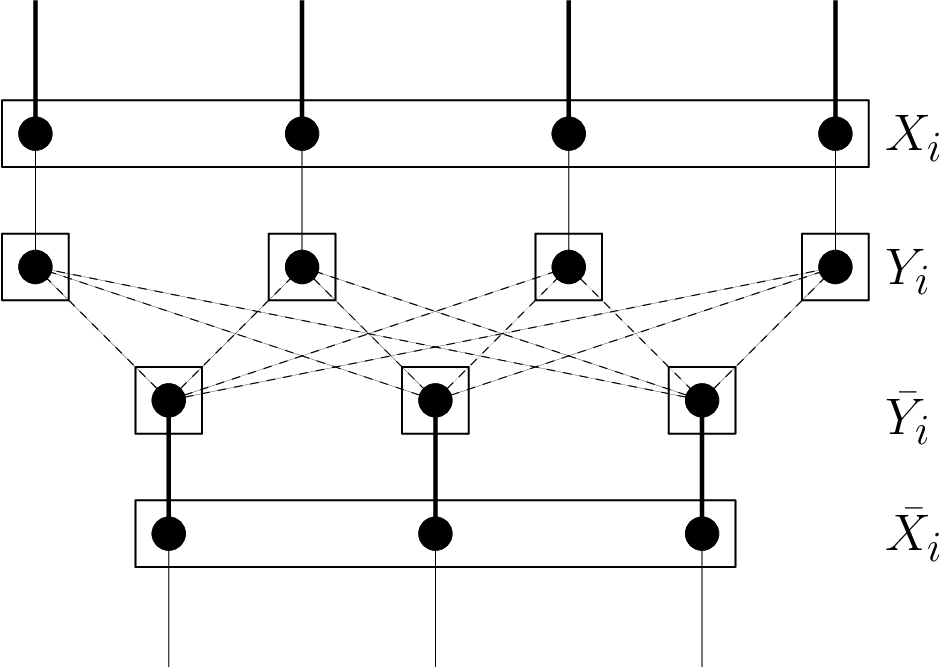}

\caption{Matching edges in the variable gadget when $x_i$ is set to true. Bold edges are in $M$. The dotted lines represent special edges.}
\label{Mvari}
\end{center}
\end{figure}

In every clause gadget, there is at least one vertex corresponding to a literal that is matched via an interconnecting edge in $M$ since every clause contains a literal set to true. Suppose that $a_1^j$ is such a vertex in the clause gadget corresponding to $c_j$ (this does not forbid that $b_1^j$ and $c_1^j$ are such vertices too). Let the edges $c_2^jc_3^j$, $b_4^jb_5^j$, $a_2^ja_3^j$, $c_4^jc_5^j$, $b_2^jb_3^j$, $a_4^ja_5^j$, $b_6^jc_6^j$ and $b_7^jc_7^j$ belong to $M$. The other two cases are symmetric.

Then, $M$ contains no forbidden edges from $F$.

Now we show for every player that it cannot belong to a strongly blocking coalition.
If the variable $x_i$ is set to true, then all the vertices in $X_i$, $\bar{X_i}$ and $\bar{Y_i}$ are matched in $M$; therefore, players $X_i$, $\bar{X_i}$ and the singleton players of $\bar{Y_i}$ cannot belong to a strongly blocking coalition. Since a vertex (which is also a player) in $Y_i$ is only connected to one of these players, it also cannot belong to a blocking coalition. If $x_i$ is set to false, a similar argument shows that these players cannot belong to a strongly blocking coalition.

Finally, similarly as in the proof of Lemma~\ref{special}, the players in the clause gadgets also cannot form a strongly blocking coalition by themselves, as at most one of them has 2 vertices unmatched, the rest have at most one.
\end{proof}
This concludes the proof of Theorem \ref{nphard}.\end{proof}

\subsection{Strong core}

We proceed to prove analogous hardness results for the strong core.

\begin{theorem} \label{thm:scmembership}
    \strongverif-3 is coNP-complete and \strongex-3 is coNP-hard. Both hardness results hold for bipartite graphs.
\end{theorem}

\begin{proof}
      We reduce from the NP-complete problem \textsc{x3c} as in the proof of Theorem \ref{thm:weakdecision}; we use the same notation as in that proof. 
    %An instance of \textsc{x3c} consist of $3n$ elements $\mathcal{X}=\{ 1,2,\dots,3n\}$ and family of 3-sets $\mathcal{S}=\{ S_1,S_2,\dots,s_{m}\}$ and the question is whether there exists a subset $\mathcal{S'}\subseteq \mathcal{S}$, such that the sets in $\mathcal{S'}$ cover each element exactly once. 
    Let $I$ be an instance of \textsc{x3c}. We create an instance $I'=(G,\partition)$ of \strongex\ and an instance $I''=(G,\partition,M)$ of \strongverif\ as follows. Note that $G$ and $\partition$ will be the same in $I'$ and $I''$.
    
    First, we describe the vertex set $V$ of the $G$.
    \begin{itemize}
        \item [--] For each element $i\in [3n]$ we create an element vertex $a_i$ and a friend vertex $x_i$, except that we do not create $x_{3n}$ (only $x_1,\dots,x_{3n-1}$).
        \item [--] We create connector vertices $b_{1,2}^1,b_{1,2}^2,b_{2,3}^2,b_{2,3}^3\dots,b_{3n-2,3n-1}^{3n-2},b_{3n-2,3n-1}^{3n-1} $ and $y_{1,2}^1,y_{1,2}^2,y_{2,3}^2, y_{2,3}^3,\dots$, $y_{3n-2,3n-1}^{3n-2},y_{3n-2,3n-1}^{3n-1}$.
        \item [--] For each set $S_j\in \mathcal{S}$, we create vertices $c_j^1,c_j^2,c_j^3$ and vertices $d_j^1,d_j^2,d_j^3$. 
        
    \end{itemize}

    We now describe the edges of the graph.
    \begin{itemize}
        \item [--] For each $j\in [m]$, $l\in [3]$ we have an edge $c_j^ld_j^l$.
        \item [--] For each $i\in [3n-1]$ we have an edge $x_ia_i$ (note that there is no vertex $x_{3n}$). 
        \item [--] For each $i\in [3n-2]$, we have edges $b_{i,i+1}^ib_{i,i+1}^{i+1},y_{i,i+1}^iy_{i,i+1}^{i+1}$.
        \item [--] For each set $S_j=\{ j_1,j_2,j_3\}$ with $j_1<j_2<j_3$ we have edges $c_j^1a_{j_1},c_j^2a_{j_2},c_j^3a_{j_3}$.
    \end{itemize}

    The graph is bipartite with bipartition $U=\{ c_1^1,c_1^2,c_1^3,\dots,c_m^1,c_m^2,c_m^2,x_1,\dots,x_{3n-1}\} \cup \{ b_{i,i+1}^i,y_{i,i+1}^i\mid i \in [3n-1]\}$ and $W=V\setminus U$.

    Finally, we describe the partition $\partition$ of the players. For each $j\in [m]$, we have a player $C_j=\{ c_j^1,c_j^2,c_j^3\}$ and a player $D_j=\{ d_j^1,d_j^2,d_j^3\}$. For each $i\in \{ 2,\dots,3n-2\}$ we have a player $A_i=\{ b_{i-1,i}^i,a_i,b_{i,i+1}^i\}$ and a player $X_i=\{ y_{i-1,i}^i,x_i,y_{i,i+1}^i\}$ Finally, we have players $A_1=\{ a_1,b_{1,2}^1\}$, $X_1=\{ x_1,y_{1,2}^1\}$ and $A_{3n-1}=\{ b_{3n-2,3n-1}^{3n-1},a_{3n-1},a_{3n}\}$, $X_{3n-1}=\{ y_{3n-2,3n-1}^{3n-1},x_{3n-1}\}$. All players have size at most 3.

    We start by showing that if there is a matching $M'$ in the strong core of $I'$, then $M'$ is the matching that contains all edges of $E$, except the ones of the form $c_j^la_{j_l}$. First, observe that $M'$ must contain all edges $c_j^ld_j^l$ for $j\in [m],l\in [3]$. Suppose that there is an edge $c_j^ld_j^l$ not in $M'$. Then, player $D_j$ does not have all of its vertices covered in $M'$. However, $C_j$ and $D_j$ can create a matching among themselves that matches each of their vertices, so $\{ C_j,D_,\}$ would be a weakly blocking coalition to $M'$, a contradiction. This implies that none of the edges of the form $c_j^la_{j_l}$ are in $M'$, as the $c_j^l$ vertices are matched elsewhere. Finally, observe that all the remaining edges can be added simultaneously, hence $M'$ must contain all of them, as any strong core matching is (inclusionwise) maximal. Note that this matching $M'$ matches every vertex of $V$, except $a_{3n}$.

    We let this matching be the matching $M$ in $I''$. 
    Therefore, by showing that the matching $M$ is in the strong core of $I''$ if and only if there is no exact 3-cover in $I$, we also show that the strong core is nonempty in $I'$ if and only if there is no exact 3-cover in $I$. 

    \begin{claim}
        If there is an exact 3-cover in $I$, then $M$ is not in the strong core in $I''$.
    \end{claim}
    \begin{proof}
        Let $\{ S_{l_1},\dots S_{l_n}\} $ be an exact 3-cover in $I$. We show that the players $C_{l_1},\dots,C_{l_n},A_1,A_2,\dots,A_{3n-1}$ form a weakly blocking coalition. As $\{ S_{l_1},\dots S_{l_n} \}$ is an exact cover, there must be a perfect matching between the included $c_j^l$ and $a_i$ vertices. Also, we can include all edges of the form $b_{i,i+1}^ib_{i,i+1}^{i+1}$ for $i=1,\dots,3n-2$. Hence, there is a perfect matching on the vertex set of these players. As $A_{3n-1}$ had a vertex exposed in $M$, it follows that $M$ is not in the strong core.
    \end{proof}
    \begin{claim}
        If $M$ is not in the strong core of $I''$, then there is an exact 3-cover in $I$.
    \end{claim}
    \begin{proof}
        Suppose that $M$ is not in the strong core. This means that there is a weakly blocking coalition $\mathcal{P}$ to $M$ with a matching $M'$ in the induced graph $G[\mathcal{P}]$. As $a_{3n}$ is the only vertex that is not covered by $M$, it follows that player $A_{3n-1}$ is in the blocking coalition and the vertex $a_{3n}$ is matched. As it can only be matched to a player $C_j$ such that $a_{3n}\in S_j$, it follows that there is a player $C_j$ in the coalition too. Furthermore, one of the vertices of $C_j$ is matched to $a_{3n}$, hence $D_j$ cannot have all its vertices covered and therefore cannot be in the coalition. As $C_j$ must weakly improve, all of its vertices must be covered by $M'$. Hence, $c_j^1a_{j_1},c_j^2a_{j_2},c_j^3a_{3n}\in M'$ and thus $A_{j_1},A_{j_2}\in \mathcal{P}$. 

        Now, we make two observations. The first is that none of the players $X_1,\dots ,X_{3n-1}$ can participate in $\mathcal{P}$. Suppose for the contrary that $X_i\in \mathcal{P}$. Then, $M'$ covers $y_{i-1,i}^i,y_{i,i+1}^i$ (whenever such a vertex exists). Hence, $X_{i-1},X_{i+1}\in \mathcal{P}$ and by iterating this argument we get that $X_i\in \mathcal{P}$ for $i\in [3n-1]$. However, $a_{j_1}$ is not matched to $x_{j_1}$, hence $x_{j_1}$ is exposed in $M'$, so player $X_{j_1}$ does not weakly improve, a contradiction. The second observation is that players $A_1,\dots,A_{3n-1}$ are all in $\mathcal{P}$. We have already seen that $A_{3n-1}\in \mathcal{P}$. Therefore, to weakly improve, $b_{3n-2,3n-1}^{3n-1}$ is matched, so $b_{3n-2,3n-1}^{3n-1}b_{3n-2,3n-1}^{3n-2}\in M'$. This implies that $A_{3n-2}\in \mathcal{P}$ and by similar arguments we get that $A_{3n-3},\dots,A_1$ are in $\mathcal{P}$.

        Finally, as each player must weakly improve, every $a_i$ vertex is matched in $M'$ and by the first observation they are matched to $C_j$ players. Whenever a $C_j$ player participates in the blocking by matching one of its vertices to an $a_i$ vertex, its corresponding $D_j$ player cannot keep all its vertices covered, so it is not in the blocking coalition. This implies that whenever $C_j\in \mathcal{P}$, $\{ c_j^1a_{j_1},c_j^2a_{j_2},c_j^3a_{j_3} \} \subseteq M'$. Combining this with the fact that each $a_i$ is matched, we obtain that these players must correspond to an exact 3-cover. 
    \end{proof}
This completes the proof of the theorem.
\end{proof}

\section{Conclusions}

We studied the NTU Partitioned Matching Game as a natural model for international kidney exchange programs (IKEPs), where players represent countries or institutions and utilities represent the number of patients matched. Our results show that in the special case where each player controls two vertices -- the NTU Matching Game with Couples case -- the weak core is always nonempty, and the existence of strong core solutions can be decided in polynomial time. Moreover, while optimizing over the weak core is NP-hard, we provide a polynomial-time algorithm to optimize over the strong core. We also prove that when the number of players is constant, all core-related decision and optimization problems are solvable in polynomial time. The latter findings may yield practically useful algorithms for IKEPs, where the number of participating countries is typically small, and they suggest core allocations as a viable and robust mechanism design goal in such settings.

Several questions remain open for future research. One direction is to explore different parameterized approaches for search and optimization in the weak or strong core (e.g. with graph-theoretic parameters), given its computational hardness. Another promising direction is to investigate dynamic or multi-round settings, reflecting how kidney exchanges are implemented in practice. Extending the model to incorporate more realistic constraints -- such as donor compatibility scores, chain exchanges, or altruistic donors -- could also bridge the gap between theory and application. %Finally, it would be valuable to identify other classes of NTU partitioned matching game beyond the ``couples'' case where core solutions remain tractable.

\section*{Acknowledgement}

The authors would like to thank Péter Biró and Zsuzsanna Jankó for the fruitful discussions on the topic, and the anonymous referees for their insightful comments.
Gergely Csáji is supported by the Hungarian Scientific Research Fund, ADVANCED grant no.~150556, by the Lend\"ulet Programme of the Hungarian Academy of Sciences -- grant number LP2021-2/2021 and by the Ministry of Culture and Innovation of Hungary from the National Research, Development and Innovation fund, financed under the KDP-2023 funding scheme (grant number C2258525).
Tamás Király is supported by the Lend\"ulet Programme of the Hungarian Academy of Sciences -- grant number LP2021-1/2021, by the Ministry of Innovation and Technology of Hungary from the National Research, Development and Innovation Fund -- ELTE TKP grant no.~2021-NKTA-62 and ADVANCED grant no.~150556.

\bibliographystyle{plain}
\bibliography{references}

@article{benedek2023computing,
  title={Computing Balanced Solutions for Large International Kidney Exchange Schemes When Cycle Length Is Unbounded},
  author={Benedek, M{\'a}rton and Bir{\'o}, P{\'e}ter and Cs{\'a}ji, Gergely and Johnson, Matthew and Paulusma, Dani{\"e}l and Ye, Xin},
  journal={The 23rd International Conference on Autonomous Agents and Multi-Agent Systems},
  year={2024}
}

@article{ashlagi2014free,
  title={Free riding and participation in large scale, multi-hospital kidney exchange},
  author={Ashlagi, Itai and Roth, Alvin E},
  journal={Theoretical Economics},
  volume={9},
  number={3},
  pages={817--863},
  year={2014},
  publisher={Wiley Online Library}
}

@article{roth2005transplant,
  title={Transplant center incentives in kidney exchange},
  author={Roth, Alvin E and S{\"o}nmez, Tayfun and Unver, M Utku},
  journal={Harvard University and Ko{\c{c}} University, unpublished mimeo},
  year={2005}
}

@article{roth1984stability,
  title={Stability and polarization of interests in job matching},
  author={Roth, Alvin E},
  journal={Econometrica: Journal of the Econometric Society},
  pages={47--57},
  year={1984},
  publisher={JSTOR}
}

@article{blair1988lattice,
  title={The lattice structure of the set of stable matchings with multiple partners},
  author={Blair, Charles},
  journal={Mathematics of operations research},
  volume={13},
  number={4},
  pages={619--628},
  year={1988},
  publisher={INFORMS}
}

@article{sotomayor1999three,
  title={Three remarks on the many-to-many stable matching problem},
  author={Sotomayor, Marilda},
  journal={Mathematical social sciences},
  volume={38},
  number={1},
  pages={55--70},
  year={1999},
  publisher={Elsevier}
}

@article{gale1962college,
  title={College admissions and the stability of marriage},
  author={Gale, David and Shapley, Lloyd S},
  journal={The American mathematical monthly},
  volume={69},
  number={1},
  pages={9--15},
  year={1962},
  publisher={Taylor \& Francis}
}

@article{irving1985efficient,
  title={An efficient algorithm for the “stable roommates” problem},
  author={Irving, Robert W},
  journal={Journal of Algorithms},
  volume={6},
  number={4},
  pages={577--595},
  year={1985},
  publisher={Elsevier}
}

@article{kaneko1982central,
  title={The central assignment game and the assignment markets},
  author={Kaneko, Mamoru},
  journal={Journal of Mathematical Economics},
  volume={10},
  number={2-3},
  pages={205--232},
  year={1982},
  publisher={Elsevier}
}

@article{scarf1967core,
  title={The core of an N person game},
  author={Scarf, Herbert E},
  journal={Econometrica: Journal of the Econometric Society},
  pages={50--69},
  year={1967},
  publisher={JSTOR}
}

@article{biro2016NTU,
  title={Fractional solutions for capacitated NTU-games, with applications to stable matchings},
  author={Bir{\'o}, P{\'e}ter and Fleiner, Tam{\'a}s},
  journal={Discrete Optimization},
  volume={22},
  pages={241--254},
  year={2016},
  publisher={Elsevier}
}

@article{biro2018stable,
  title={The stable fixtures problem with payments},
  author={Bir{\'o}, P{\'e}ter and Kern, Walter and Paulusma, Dani{\"e}l and Wojuteczky, P{\'e}ter},
  journal={Games and economic behavior},
  volume={108},
  pages={245--268},
  year={2018},
  publisher={Elsevier}
}

@article{shapley1971assignment,
  title={The assignment game I: The core},
  author={Shapley, Lloyd S and Shubik, Martin},
  journal={International Journal of game theory},
  volume={1},
  number={1},
  pages={111--130},
  year={1971},
  publisher={Springer}
}

@incollection{sotomayor1992multiple,
  title={The multiple partners game},
  author={Sotomayor, Marilda},
  booktitle={Equilibrium and Dynamics: Essays in Honour of David Gale},
  pages={322--354},
  year={1992},
  publisher={Springer}
}

@article{csajiunbounded,
  title={Computing balanced solutions for large international kidney exchange schemes when cycle length is unbounded},
  author={Benedek, M{\'a}rton and Bir{\'o}, P{\'e}ter and Cs{\'a}ji, Gergely and Johnson, Matthew and Paulusma, Dani{\"e}l and Ye, Xin},
  journal={arXiv preprint arXiv:2312.16653},
  year={2023}
}

@article{druzsin2024performance,
  title={Performance evaluation of national and international kidney exchange programmes with the ENCKEP simulator},
  author={Druzsin, Krist{\'o}f and Bir{\'o}, P{\'e}ter and Klimentova, Xenia and Fleiner, Rita},
  journal={Central European Journal of Operations Research},
  volume={32},
  number={4},
  pages={923--943},
  year={2024},
  publisher={Springer}
}

@article{sonmez2013market,
  title={Market design for kidney exchange},
  author={S{\"o}nmez, Tayfun and {\"U}nver, M Utku},
  journal={The Handbook of Market Design},
  volume={93137},
  year={2013}
}

@article{ashlagi2015mix,
  title={Mix and match: A strategyproof mechanism for multi-hospital kidney exchange},
  author={Ashlagi, Itai and Fischer, Felix and Kash, Ian A and Procaccia, Ariel D},
  journal={Games and Economic Behavior},
  volume={91},
  pages={284--296},
  year={2015},
  publisher={Elsevier}
}

@article{benedek2025partitioned,
  title={Partitioned matching games for international kidney exchange},
  author={Benedek, M{\'a}rton and Bir{\'o}, P{\'e}ter and Kern, Walter and P{\'a}lv{\"o}lgyi, D{\"o}m{\"o}t{\"o}r and Paulusma, Daniel},
  journal={Mathematical Programming},
  pages={1--36},
  year={2025},
  publisher={Springer}
}

@inproceedings{blum1990new,
  title={A new approach to maximum matching in general graphs},
  author={Blum, Norbert},
  booktitle={Automata, Languages and Programming: 17th International Colloquium Warwick University, England, July 16--20, 1990 Proceedings 17},
  pages={586--597},
  year={1990},
  organization={Springer}
}

@article{benedek2021computing,
  title={Computing balanced solutions for large international kidney exchange schemes},
  author={Benedek, M{\'a}rton and Bir{\'o}, P{\'e}ter and Paulusma, Daniel and Ye, Xin},
  journal={arXiv preprint arXiv:2109.06788},
  year={2021}
}

@inproceedings{kern2019generalized,
  title={Generalized matching games for international kidney exchange},
  author={Kern, Walter and Bir{\'o}, P{\'e}ter and P{\'a}lv{\"o}lgyi, D{\"o}m{\"o}t{\"o}r and Paulusma, Daniel},
  booktitle={18th International Conference on Autonomous Agents and Multiagent Systems, AAMAS 2019},
  pages={413--421},
  year={2019},
  organization={Association for Computing Machinery}
}

@article{agarwal2019market,
  title={Market failure in kidney exchange},
  author={Agarwal, Nikhil and Ashlagi, Itai and Azevedo, Eduardo and Featherstone, Clayton R and Karaduman, {\"O}mer},
  journal={American Economic Review},
  volume={109},
  number={11},
  pages={4026--4070},
  year={2019},
  publisher={American Economic Association 2014 Broadway, Suite 305, Nashville, TN 37203}
}

@article{williamson2016fast,
  title={Fast algorithms for the undirected negative cost cycle detection problem},
  author={Williamson, Matthew and Eirinakis, Pavlos and Subramani, K},
  journal={Algorithmica},
  volume={74},
  number={1},
  pages={270--325},
  year={2016},
  publisher={Springer}
}

@article{roth2005pairwise,
  title={Pairwise kidney exchange},
  author={Roth, Alvin E and S{\"o}nmez, Tayfun and {\"U}nver, M Utku},
  journal={Journal of Economic theory},
  volume={125},
  number={2},
  pages={151--188},
  year={2005},
  publisher={Elsevier}
}

@article{biro2021modelling,
  title={Modelling and optimisation in European kidney exchange programmes},
  author={Bir{\'o}, P{\'e}ter and Van de Klundert, Joris and Manlove, David and Pettersson, William and Andersson, Tommy and Burnapp, Lisa and Chromy, Pavel and Delgado, Pablo and Dworczak, Piotr and Haase, Bernadette and others},
  journal={European Journal of Operational Research},
  volume={291},
  number={2},
  pages={447--456},
  year={2021},
  publisher={Elsevier}
}

@article{biro2024strong,
  title={Strong core and Pareto-optimality in the multiple partners matching problem under lexicographic preference domains},
  author={Bir{\'o}, P{\'e}ter and Cs{\'a}ji, Gergely},
  journal={Games and Economic Behavior},
  volume={145},
  pages={217--238},
  year={2024},
  publisher={Elsevier}
}

@article{klimentova2021fairness,
  title={Fairness models for multi-agent kidney exchange programmes},
  author={Klimentova, Xenia and Viana, Ana and Pedroso, Jo{\~a}o Pedro and Santos, Nicolau},
  journal={Omega},
  volume={102},
  pages={102333},
  year={2021},
  publisher={Elsevier}
}

@inproceedings{ashlagi2011individual,
  title={Individual rationality and participation in large scale, multi-hospital kidney exchange},
  author={Ashlagi, Itai and Roth, Alvin},
  booktitle={Proceedings of the 12th ACM conference on Electronic commerce},
  pages={321--322},
  year={2011}
}

@article{ashlagi2012new,
  title={New challenges in multihospital kidney exchange},
  author={Ashlagi, Itai and Roth, Alvin E},
  journal={American Economic Review},
  volume={102},
  number={3},
  pages={354--359},
  year={2012},
  publisher={American Economic Association}
}

@article{carvalho2017nash,
  title={Nash equilibria in the two-player kidney exchange game},
  author={Carvalho, Margarida and Lodi, Andrea and Pedroso, Joao Pedro and Viana, Ana},
  journal={Mathematical Programming},
  volume={161},
  pages={389--417},
  year={2017},
  publisher={Springer}
}

@article{carvalho2023Nash,
  title={A theoretical and computational equilibria analysis of a multi-player kidney exchange program},
  author={Carvalho, Margarida and Lodi, Andrea},
  journal={European Journal of Operational Research},
  volume={305},
  number={1},
  pages={373--385},
  year={2023},
  publisher={Elsevier}
}

@inproceedings{gourves2008cooperation,
  title={Cooperation in multiorganization matching},
  author={Gourves, Laurent and Monnot, J{\'e}r{\^o}me and Pascual, Fanny},
  booktitle={International Workshop on Approximation and Online Algorithms},
  pages={78--91},
  year={2008},
  organization={Springer}
}

@inproceedings{hajaj2015strategy,
  title={Strategy-proof and efficient kidney exchange using a credit mechanism},
  author={Hajaj, Chen and Dickerson, John and Hassidim, Avinatan and Sandholm, Tuomas and Sarne, David},
  booktitle={Proceedings of the AAAI Conference on Artificial Intelligence},
  volume={29},
  year={2015}
}

@article{biro2019building,
  title={Building kidney exchange programmes in Europe—an overview of exchange practice and activities},
  author={Bir{\'o}, P{\'e}ter and Haase-Kromwijk, Bernadette and Andersson, Tommy and {\'A}sgeirsson, Eyj{\'o}lfur Ingi and Baltesov{\'a}, Tatiana and Boletis, Ioannis and Bolotinha, Catarina and Bond, Gregor and B{\"o}hmig, Georg and Burnapp, Lisa and others},
  journal={Transplantation},
  volume={103},
  number={7},
  pages={1514--1522},
  year={2019},
  publisher={LWW}
}

@article{bohmig2017czech,
  title={Czech-Austrian kidney paired donation: first European cross-border living donor kidney exchange},
  author={B{\"o}hmig, Georg A and Fronek, Ji{\v{r}}{\'\i} and Slavcev, Antonij and Fischer, Gottfried F and Berlakovich, Gabriela and Viklicky, Ondrej},
  journal={Transplant International},
  volume={30},
  number={6},
  pages={638--639},
  year={2017}
}

@article{valentin2019international,
  title={International cooperation for kidney exchange success},
  author={Valent{\'\i}n, Mar{\'\i}a Oliva and Garcia, Marta and Costa, Alessandro Nanni and Bolotinha, Catarina and Guirado, Lluis and Vistoli, Fabio and Breda, Alberto and Fiaschetti, Pamela and Dominguez-Gil, Beatriz},
  journal={Transplantation},
  volume={103},
  number={6},
  pages={e180--e181},
  year={2019},
  publisher={LWW}
}

@mastersthesis{collette2023weak,
  title={Weak core solution for the non-transferable utility kidney exchange game},
  author={Collette, Rapha{\"e}l},
  year={2023},
  school={Université de Montréal}
}

@article{saidman2006increasing,
  title={Increasing the opportunity of live kidney donation by matching for two-and three-way exchanges},
  author={Saidman, Susan L and Roth, Alvin E and S{\"o}nmez, Tayfun and {\"U}nver, M Utku and Delmonico, Francis L},
  journal={Transplantation},
  volume={81},
  number={5},
  pages={773--782},
  year={2006},
  publisher={LWW}
}

@techreport{egres-19-12,
AUTHOR	=	{Kir{\'a}ly, Tam{\'a}s and M{\'e}sz{\'a}ros-Karkus, Zsuzsa},
TITLE	=	{Complexity of the NTU International Matching Game},
NOTE	=	{{\tt egres.elte.hu}},
INSTITUTION	=	{Egerv{\'a}ry Research Group, Budapest},
YEAR	=	{2019},
NUMBER	=	{TR-2019-12}
}

@article{blom2022rejection,
  title={Rejection-proof kidney exchange mechanisms},
  author={Blom, Danny and Smeulders, Bart and Spieksma, Frits CR},
  journal={arXiv preprint arXiv:2206.11525},
  year={2022}
}

@article{csaji2023couples,
  title={Couples can be tractable: New algorithms and hardness results for the Hospitals/Residents problem with Couples},
  author={Cs{\'a}ji, Gergely and Manlove, David and McBride, Iain and Trimble, James},
  journal={arXiv preprint arXiv:2311.00405},
  year={2023}
}

@article{aldershof1996stable,
  title={Stable matchings with couples},
  author={Aldershof, Brian and Carducci, Olivia M},
  journal={Discrete Applied Mathematics},
  volume={68},
  number={1-2},
  pages={203--207},
  year={1996},
  publisher={Elsevier}
}

@article{ronn1990np,
  title={NP-complete stable matching problems},
  author={Ronn, Eytan},
  journal={Journal of Algorithms},
  volume={11},
  number={2},
  pages={285--304},
  year={1990},
  publisher={Elsevier}
}

@article{biro2016matching,
  title={Matching couples with Scarf’s algorithm},
  author={Bir{\'o}, P{\'e}ter and Fleiner, Tam{\'a}s and Irving, Robert W},
  journal={Annals of Mathematics and Artificial Intelligence},
  volume={77},
  pages={303--316},
  year={2016},
  publisher={Springer}
}

@article{nguyen2018near,
  title={Near-feasible stable matchings with couples},
  author={Nguyen, Thanh and Vohra, Rakesh},
  journal={American Economic Review},
  volume={108},
  number={11},
  pages={3154--3169},
  year={2018},
  publisher={American Economic Association 2014 Broadway, Suite 305, Nashville, TN 37203}
}

@article{mcbride2013hospitals,
  title={The hospitals/residents problem with couples: Complexity and integer programming models},
  author={McBride, Iain and Manlove, David F},
  journal={CoRR, abs/1308.4534},
  year={2013}
}

@article{klaus2005stable,
  title={Stable matchings and preferences of couples},
  author={Klaus, Bettina and Klijn, Flip},
  journal={Journal of Economic Theory},
  volume={121},
  number={1},
  pages={75--106},
  year={2005},
  publisher={Elsevier}
}

@article{marx2011stable,
  title={Stable assignment with couples: Parameterized complexity and local search},
  author={Marx, D{\'a}niel and Schlotter, Ildik{\'o}},
  journal={Discrete Optimization},
  volume={8},
  number={1},
  pages={25--40},
  year={2011},
  publisher={Elsevier}
}

@article{roth2006utilizing,
  title={Utilizing list exchange and nondirected donation through ‘chain’paired kidney donations},
  author={Roth, Alvin E and S{\"o}nmez, Tayfun and {\"U}nver, M Utku and Delmonico, Francis L and Saidman, Susan L},
  journal={American Journal of transplantation},
  volume={6},
  number={11},
  pages={2694--2705},
  year={2006},
  publisher={Elsevier}
}

\end{document}